\documentclass[letterpaper, 10 pt,journal]{IEEEtran}

 \IEEEpubidadjcol 
\IEEEoverridecommandlockouts

\usepackage{graphics} 
\usepackage{graphicx}
\usepackage{epsfig} 
\usepackage{amsmath} 
\usepackage{amssymb}  

\usepackage{tikz} 
\usepackage{pgfplots} 

\usepackage{algpseudocode}
\usepackage{algorithm}
\usepackage{epstopdf}

\usepackage{color}

\usepackage{lipsum}
\usepackage{verbatim}

\usepackage{amsthm}

\usepackage{epstopdf}

\usepackage{mathtools}

\newtheorem{assumption}{\bf Assumption}

\newtheorem{theorem}{\bf Theorem}
\newtheorem{proposition}{\bf Proposition}
\newtheorem{lemma}{\bf Lemma}

\newtheorem{remark}{\bf Remark}

\usepackage{hyperref}

\title{
Linear robust adaptive model predictive control:  \\ 
Computational complexity and conservatism \\- extended version
}
\author{Johannes K\"ohler$^1$, Elisa Andina$^2$, Raffaele Soloperto$^1$, Matthias A. M\"uller$^3$, Frank Allg\"ower$^1$
\thanks{%
This work was supported by the German Research Foundation under Grants GRK 2198/1 - 277536708, AL 316/12-2, and MU 3929/1-2 - 279734922. 
The authors thank the International Max Planck Research School for Intelligent Systems (IMPRS-IS) for supporting Raffaele Soloperto.
}
\thanks{$^1$Johannes K\"ohler, Raffaele Soloperto and Frank Allg\"ower are with the Institute for Systems Theory and Automatic Control, University of Stuttgart, 70550 Stuttgart, Germany.
(email:$\{$johannes.koehler, raffaele.soloperto, frank.allgower$\}$@ist.uni-stuttgart.de).}
\thanks{$^2$Elisa Andina is a M.Sc. student at the Universit\`{a} di Bologna, 40136 Bologna, Italiy (email: elisa.andina@studio.unibo.it).}
\thanks{$^3$Matthias A. M\"uller is with the Institute of Automatic Control, Leibniz University Hannover, 30167 Hannover, Germany.
(email:mueller@irt.uni-hannover.de).}
}
\begin{document}
\IEEEoverridecommandlockouts
\IEEEpubid{\begin{minipage}{\textwidth}\ \\[12pt] \\ \\
         \copyright 2019 IEEE.  Personal use of this material is  permitted.  Permission from IEEE must be obtained for all other uses, in  any current or future media, including reprinting/republishing this material for advertising or promotional purposes, creating new  collective works, for resale or redistribution to servers or lists, or  reuse of any copyrighted component of this work in other works.
     \end{minipage}}

\maketitle

\begin{abstract}
In this paper, we present a robust adaptive model predictive control (MPC) scheme for linear systems subject to parametric uncertainty and additive disturbances. 
The proposed approach provides a computationally efficient formulation with theoretical guarantees (constraint satisfaction and stability), while allowing for reduced conservatism and improved performance due to online parameter adaptation.
A moving window parameter set identification is used to compute a fixed complexity parameter set based on past data. 
Robust constraint satisfaction is achieved by using a computationally efficient tube based robust MPC method. 
The predicted cost function is based on a least mean squares point estimate, which ensures finite-gain $\mathcal{L}_2$ stability of the closed loop. 
The overall algorithm has a fixed (user specified) computational complexity. 
We illustrate the applicability of the approach and the trade-off between conservatism and computational complexity using a numerical example.

This paper is an extended version of~\cite{Koehler2019Adaptive}, and contains additional details regarding the theoretical proof of Theorem~\ref{thm:main}, the numerical example (Sec.~\ref{sec:num}),  and the offline computations in Appendix~\ref{App:tube}--\ref{App:terminal}.  
\end{abstract}

\section{Introduction}
\subsubsection*{Motivation}
Model predictive control (MPC)~\cite{rawlings2017model,kouvaritakis2016model} is an optimization based control method that can handle complex multi-input multi-output (MIMO) systems with hard state and input constraints. 
Model uncertainty and disturbances can have an adversarial impact on performance and constraint satisfaction. 
These issues can be addressed using robust MPC formulations, which can, however, be quite conservative in case of parametric uncertainty. 
This motivates the design of adaptive/learning MPC schemes that can use online model adaptation to reduce the conservatism and improve the performance. 
In this paper, we provide a robust adaptive MPC formulation that provides theoretical guarantees, allows for online performance improvement and is computationally efficient. 

\subsubsection*{Related work}
Tube based robust MPC methods are the simplest way to ensure stability and robust constraint satisfaction under uncertainty. 
Robust constraint satisfaction is typically achieved by including a pre-stabilizing feedback and computing a polytopic tube that bounds the uncertain predicted trajectories, compare for example~\cite{fleming2015robust,Robust_TAC_19,rakovic2012homothetic}.  
The performance and conservatism of these robust MPC approaches can be improved by using online system identification or parameter adaptation.

In~\cite{aswani2013provably,di2016indirect}, a fixed uncertainty description is used to ensure robust constraint satisfaction, while an online adapted model is used in the cost function to improve performance. 

The conservatism of the robust MPC constraint tightening can be reduced by using online set-membership system identification. 
For FIR systems, this has been proposed in~\cite{tanaskovic2014adaptive} and extended to time-varying systems and stochastic uncertainty in~\cite{tanaskovic2019adaptive} and~\cite{bujarbaruah2018adaptive}, respectively. 
In~\cite{lorenzen2019robust}, this approach has been extended to general linear state space models, by combining set-membership estimation with the homothetic tube based robust MPC formulation in~\cite{rakovic2012homothetic}. 
By using an additional least mean square (LMS) point estimate for the cost function, this scheme also ensures finite-gain $\mathcal{L}_2$-stability. 
The main drawback of this approach is that the resulting quadratic program (QP) has significantly more optimization variables and constraints compared to a nominal MPC.

Robust adaptive MPC methods for nonlinear systems can be found in~\cite{adetola2011robust}. 
Methods to explicitly incentivize or enforce learning of the unknown parameters can be found in~\cite{marafioti2014persistently,heirung2017dual,mesbah2018stochastic,RS_CDC_Dual_19}, which is, however, not the goal of this paper.

\IEEEpubidadjcol

\subsubsection*{Contribution}
In this work, we present a computationally efficient robust adaptive MPC scheme for linear uncertain systems. 
Similar to~\cite{lorenzen2019robust}, the proposed method uses a set-membership estimate for the parametric uncertainty, a robust constraint tightening, and a cost function based on a LMS point estimate. 
Correspondingly, the proposed robust adaptive MPC scheme shares the theoretical properties of the scheme in~\cite{lorenzen2019robust}, i.e.,  ensures robust constraint satisfaction and finite-gain $\mathcal{L}_2$ stability w.r.t. additive disturbances. 
Compared to~\cite{lorenzen2019robust}, we employ a moving window set-membership estimation to compute a hypercube that bounds the parametric uncertainty. 
Furthermore, we use a robust constraint tightening based on the novel robust MPC framework in~\cite{Robust_TAC_19}. 
As a result, the proposed  robust adaptive MPC scheme has a fixed computationally complexity, which is only moderately increased compared to a nominal MPC scheme. 
The main theoretical contribution of this paper is to extend the robust MPC method presented in~\cite{Robust_TAC_19} to allow for an online adaptation of the parametric uncertainty, while preserving the theoretical properties and computational efficiency.  
Compared to~\cite{lorenzen2019robust},  the proposed approach can be more conservative. On the other hand, it has a strongly reduced computational complexity. 
The trade off between computational complexity and conservatism in different robust MPC approaches is discussed in detail and is quantitatively investigated using a numerical example. 
The numerical example demonstrates the benefits of the proposed adaptive formulation compared to a robust formulation. 
Furthermore, the computational efficiency of the proposed formulation in comparison to~\cite{lorenzen2019robust} is substantiated. 
After submission of this manuscript, a competing approach  for robust adaptive MPC has been presented in~\cite{Lu2019RAMPC} based on the robust approach in~\cite{fleming2015robust}.
In terms of computational complexity and conservatism, the approach in~\cite{Lu2019RAMPC} is located between the proposed approach and the approach in~\cite{lorenzen2019robust}, which is demonstrated in the numerical comparison.

\subsubsection*{Outline}
Section~\ref{sec:setup} presents the problem setup and discusses the parameter estimation. 
Section~\ref{sec:main} presents the proposed robust adaptive MPC scheme with corresponding theoretical guarantees and a discussion regarding alternative robust MPC approaches.
Section~\ref{sec:num} demonstrates the benefits of the proposed formulation with a numerical example. 
Section~\ref{sec:sum} concludes the paper.
The Appendix contains additional details. 
The results in this paper are based on the thesis~\cite{Elisa}.
This paper is an extended version of the conference paper~\cite{Koehler2019Adaptive}, containing detailed proofs and additional results.

\subsubsection*{Notation}
The quadratic norm with respect to a positive definite matrix $Q=Q^\top$ is denoted by $\|x\|_Q^2=x^\top Q x$. 
Denote the unit hypercube by $\mathbb{B}_p:=\{\theta\in\mathbb{R}^p|~\|\theta\|_{\infty}\leq 0.5\}$.

\section{Problem setup and parameter estimation}
\label{sec:setup}
This section introduces the problem setup (Sec.~\ref{sec:setup_1}), the computation of set-membership estimates (Sec.~\ref{sec:setup_2}) and point estimates (Sec.~\ref{sec:setup_3}) for unknown constant parameters. 
\subsection{Setup}
\label{sec:setup_1}
We consider a discrete-time linear system
\begin{align}
\label{eq:sys}
x_{t+1}=A_{\theta}x_t+B_{\theta}u_t+d_t,
\end{align}
with state $x_t\in\mathbb{R}^n$, input $u_t\in\mathbb{R}^m$, additive disturbances $d_t\in\mathbb{R}^n$ and unknown but constant parameters $\theta=\theta^*\in\mathbb{R}^p$. 
\begin{assumption}
\label{ass:model}
The disturbance satisfies $d_t\in\mathbb{D}$ for all $t\geq 0$, with $\mathbb{D}=\{d\in\mathbb{R}^n|~H_d d\leq h_d\}$ compact.
 The system matrices depend affinely on the parameters $\theta\in\mathbb{R}^p$
\begin{align*}
A_{\theta}=A_0+\sum_{i=1}^p [\theta]_i A_i,\quad
B_{\theta}=B_0+\sum_{i=1}^p [\theta]_i B_i.
\end{align*}
There exists a known prior parameter set $\Theta_0^{HC}:=\overline{\theta}_0\oplus\eta_0\mathbb{B}_p$, $\eta_0\geq 0$  that contains the parameters $\theta^*$.
\end{assumption}
These conditions are rather standard in the context of linear uncertain systems. 
For simplicity, we consider $\Theta^{HC}_0$  to be a hypercube, however, any polytope of fixed shape can be employed. 
Extensions to time-varying parameters are discussed in Section~\ref{sec:discuss}. 
We consider mixed constraints on the state and input
\begin{align}
\label{eq:con}
(x_t,u_t)\in\mathcal{Z}, ~\forall t\geq 0,
\end{align}
with the compact polytope 
\begin{align*}
\mathcal{Z}=\{(x,u)\in\mathbb{R}^{n+m}|~F_jx+G_ju\leq 1,~ j=1,\dots, q\}.
\end{align*}
The control goal is to stabilize the origin, while satisfying the constraints~\eqref{eq:con} despite the disturbances $d_t$ and the uncertainty in the parameters $\theta$. 
\subsection{Set membership estimation}
\label{sec:setup_2}
In order to satisfy the constraints~\eqref{eq:con}, the uncertainty in the parameters $\theta$ needs to be taken into account, which in turn may lead to conservatism.  
To reduce this conservatism, a set membership estimation algorithm is used to compute a smaller set $\Theta_t^{HC}\subseteq\Theta_0^{HC}$ that contains the true parameters $\theta^*$, as done in~\cite{tanaskovic2014adaptive,lorenzen2019robust,Lu2019RAMPC}.   
%
Define
\begin{align*}
D(x,u):=[A_1x+B_1u,\dots, A_px+B_pu]\in\mathbb{R}^{n\times p}. 
\end{align*}
Given $(x_{t-1},u_{t-1},x_t)$, the non falsified parameter set is given by the polytope
\begin{align*}
\Delta_t:=&\{\theta\in\mathbb{R}^p|~x_t-A_{\theta}x_{t-1}-B_{\theta}u_{t-1} \in\mathbb{D}\}.
\end{align*}
The following algorithm uses  a given hypercube $\Theta_{t-1}^{HC}$ and the past $M\in\mathbb{N}$ sets $\Delta_{t-k}$ in a moving window fashion to compute a tighter non falsified hypercube $\Theta_t^{HC}$ using linear programming (LP). 
\begin{algorithm}[h]
\caption{Moving window hypercube update} 
\label{alg:HC}
Input: $\{\Delta_{k}\}_{k=t,\dots,t-M-1}$, $\Theta_{t-1}^{HC}$. Output: $\Theta_t^{HC}=\overline{\theta}_t\oplus\eta_t\mathbb{B}_p$
\begin{algorithmic}
\State Define polytope $\Theta_t^M:=\Theta_{t-1}^{HC}\bigcap_{k=t-M-1}^{t}\Delta_k$.
\State Solve $2 p$ LPs ($i=1,\dots,p$):  
\Statex ~$\theta_{i,t,\min}:=\min_{\theta\in\Theta_t^M}e_i^\top \theta$, $\theta_{i,t,\max}:=\max_{\theta\in\Theta_t^M}e_i^\top \theta$,
\Statex ~\text{with unit vector }$e_i=[0,\dots,1,\dots,0]\in\mathbb{R}^p,~[e_i]_i=1$.
\State Set $[\overline{\theta}_t]_i=0.5(\theta_{i,t,\min}+\theta_{i,t,\max})$.
\State Set $\eta_t=\max_i(\theta_{i,t,\max}-\theta_{i,t,\min})$.
\State Project: $\overline{\theta}_t$ on $\overline{\theta}_{t-1}\oplus (\eta_{t-1}-\eta_t)\mathbb{B}_p$.
\end{algorithmic}
\end{algorithm}
\begin{lemma}
\label{lemma:HC}
Let Assumption~\ref{ass:model} hold. 
For $t\geq 0$ the recursively updated sets $\Theta_t^{HC}$, $\Theta_t^M$ contain the true parameters $\theta^*$ and satisfy
$\Theta_t^M\subseteq\Theta_t^{HC}\subseteq\Theta_{t-1}^{HC}$.
\end{lemma} 
\begin{proof}
We define the \textit{unique tight} hyperbox overapproximation of $\Theta_t^M$ as $\Theta_t^{HB}=\{\theta\in\mathbb{R}^p|~[\theta]_i\in[\theta_{i,t,\min},\theta_{i,t,\max}]\}$, which satisfies $\Theta_t^M\subseteq\Theta_t^{HB}$. 
Given that by definition $\Theta_t^M\subseteq\Theta_{t-1}^{HC}$ and the set $\Theta_t^{HB}$ is the \textit{unique} tight hyperbox overapproximation of $\Theta_t^M$, we have $\Theta_{t}^{HB}\subseteq\Theta_{t-1}^{HC}$ and thus $\eta_{t-1}\geq  \eta_t$. \\
\textbf{Part I. }The projection ensures 
\begin{align*}
\Theta_t^{HC}=&\overline{\theta}_t\oplus\eta_t\mathbb{B}_p\\
\subseteq& \overline{\theta}_{t-1}\oplus(\eta_{t-1}-\eta_{t})\mathbb{B}_p\oplus\eta_{t}\mathbb{B}_p=\Theta_{t-1}^{HC}.
\end{align*}
\textbf{Part II. }In the following, we show $\Theta_t^{HC}\supseteq\Theta_t^{HB}\supseteq\Theta_t^M$. 
This claim reduces to $[\theta_{i,t,\min},\theta_{i,t,\max}]\subseteq[\overline{\theta}]_i\oplus\eta_t[-0.5,0.5]$.
Without loss of generality, we only show $\theta_{i,t,\min}\geq [\overline{\theta}]_i-0.5\eta_t$, the case $\theta_{i,t,\max}\leq [\overline{\theta}]_i+0.5\eta_t$ is analogous.\\
\textbf{Case (i): }Suppose $0.5(\theta_{i,t,\min}+\theta_{i,t,\max})<[\overline{\theta}_{t-1}]_i-0.5(\eta_{t-1}-\eta_t)$ and thus
$[\overline{\theta}_t]_i=[\overline{\theta}_{t-1}]_i-0.5(\eta_{t-1}-\eta_t)$ due to the projection. 
Given that  $\Theta_{t}^{HB}\subseteq\Theta_{t-1}^{HC}$, we have 
$ \theta_{i,t,\min} \geq [\overline{\theta}_{t-1}]_i-0.5\eta_{t-1}=[\overline{\theta}_t]_i-0.5\eta_t$. \\
\textbf{Case (ii): }Suppose $0.5(\theta_{i,t,\min}+\theta_{i,t,\max})\geq [\overline{\theta}_{t-1}]_i-0.5(\eta_{t-1}-\eta_t)$, which implies $[\overline{\theta}_t]_i\leq 0.5(\theta_{i,t,\min}+\theta_{i,t,\max})$ (since the projection on the lower bound $\overline{\theta}_{t-1}-0.5(\eta_{t-1}-\eta_{t})$ is not active). 
Thus we have $\theta_{i,t,\min}=0.5(\theta_{i,t,\min}+\theta_{i,t,\max})-0.5(\theta_{i,t,\max}-\theta_{i,t,\min})\geq [\overline{\theta}_t]_i-0.5\eta_t$. \\
Combining both cases yields $\theta_{i,t,\min}\geq [\overline{\theta}_t]_i-0.5\eta_t$. \\
\textbf{Part III. } Assumption~\ref{ass:model} ensurs that $\theta^*\in\Delta_t$ for all $t\geq 0$. 
Suppose $\theta^*\in\Theta_t^{HC}$, then $\theta^*\in\Theta_{t+1}^M\subseteq\Theta_{t+1}^{HC}$. 
Thus, $\theta^*\in\Theta_0^{HC}$ (Ass.~\ref{ass:model}) ensures $\theta^*\in\Theta_t^M\subseteq\Theta_t ^{HC}$.
\end{proof}	
Algorithm~\ref{alg:HC} first computes the \textit{unique} (tight) hyperbox overapproximation ($\theta_{i,t,\min},\theta_{i,t,\max}$) and then an overapproximating hypercube. 
Since the overapproximating hypercube is not unique, the final projection is necessary to ensure $\Theta_t^{HC}\subseteq\Theta_{t-1}^{HC}$. 
Neglecting the scalar additions and the projection, this algorithm requires the solution of $2p$ LPs with $p$ optimization variables each. 
Thus, the computational demand of Algorithm~\ref{alg:HC} is typically small compared to the MPC optimization problem~\eqref{eq:MPC}.

Similar ideas for parameter estimation based on a moving window are also discussed in \cite[Remark~4]{lorenzen2019robust} and  \cite{chisci1998block}.
For comparison, in~\cite{ lorenzen2019robust}, the recursive update $\Theta_{t}=\Theta_{t-1}\cap\Delta_{t}$ (without any overapproximation) is used, which satisfies the same properties, compare~\cite[Lemma~2]{ lorenzen2019robust}. 
However, this corresponds to a full information filter, that has an online increasing (potentially unbounded) complexity, while in the proposed approach all operations have a fixed complexity. 
We use a hypercube to reduce the computational complexity of the robust adaptive MPC in Section~\ref{sec:main}, although in principle any polytope of fixed shape can be used.  
For general polytopes an overapproximation of fixed shape can be computed by replacing Algorithm~\ref{alg:HC} with more involved LPs, compare~\cite[Lemma~3]{ lorenzen2019robust}, \cite[Lemma~5]{Lu2019RAMPC}.

\subsection{Point estimate}
\label{sec:setup_3}
In order to improve the performance of the closed-loop system, the predicted cost function is evaluated based on a point estimate $\hat{\theta}$, as done in~\cite{aswani2013provably,di2016indirect, lorenzen2019robust}. 
In particular, we consider a least mean squares (LMS) point estimate with projection on the current parameter set. 
Given a point estimate $\hat{\theta}_{t}$, define the predicted state by $\hat{x}_{1|t}=A_{\hat{\theta}_{t}}x_t+B_{\hat{\theta}_{t}}u_t$ and the prediction error $\tilde{x}_{1|t}=x_{t+1}-\hat{x}_{1|t}$. 
The recursive LMS update is given by
\begin{align}
\label{eq:theta_LMS}
\hat{\theta}_t=\Pi_{\Theta_t^{HC}}(\hat{\theta}_{t-1}+\mu D(x_{t-1},u_{t-1})^\top \tilde{x}_{1|t-1}),
\end{align}
where $\Pi_{\Theta_t^{HC}}(\theta)=\arg\min_{\overline{\theta}\in\Theta_t^{HC}}\|\theta-\overline{\theta}\|$ denotes the Euclidean projection on the hypercube $\Theta_t^{HC}$ and $\mu>0$ is an update gain.
The update can equally be projected on the set $\Theta_t^M$, at the expense of additional computational complexity. 
\begin{lemma}
\label{lemma:LMS}
\cite[Lemma~5]{ lorenzen2019robust}
Let Assumption~\ref{ass:model} hold. 
Suppose that the parameter gain $\mu>0$ satisfies ${1}/{\mu}<\sup_{(x,u)\in\mathcal{Z}}\|D(x,u)\|^2$ and that the state and input satisfy $(x_t,u_t)\in\mathcal{Z}$ for all $t\geq 0$. 
Then for any initial parameter estimate $\hat{\theta}_0\in\Theta^{HC}_0$ and any time $k\in\mathbb{N}$, the parameter estimate~\eqref{eq:theta_LMS} satisfies
\begin{align*}
\sum_{t=0}^{k}\|\tilde{x}_{1|t}\|^2\leq \dfrac{1}{\mu}\|\hat{\theta}_0-\theta^*\|^2+\sum_{t=0}^{k}\|d_t\|^2. 
\end{align*}
\end{lemma}
\begin{proof}
This property is actually equivalent to the result in~\cite[Lemma~5]{ lorenzen2019robust}.
The only difference is that we project the point estimate on $\Theta_t^{HC}$. 
Lemma~\ref{lemma:HC} ensures that the true parameters satisfy $\theta^*\in\Theta_t^{HC}$. 
As a result, the projection does not increase the parameter estimation error (non-expansiveness of the projection operator). 
The remainder of the proof is analogous to~\cite[Lemma~5]{ lorenzen2019robust}.
\end{proof}
Alternative approaches to compute such a point estimate include Kalman filter and recursive least squares (RLS), compare~\cite{lennart1999system}. 

\section{Robust adaptive MPC}
\label{sec:main}
This section contains the proposed robust adaptive MPC approach based on the parameter estimates in Sections~\ref{sec:setup_2} and \ref{sec:setup_3}.  
Section~\ref{sec:main_1} discusses some preliminaries regarding polytopic tubes.
The proposed MPC scheme is presented in Section~\ref{sec:main_2} and  the overall online and offline computations are summarized.
The theoretical analysis is detailed in Section~\ref{sec:main_3}.
Section~\ref{sec:discuss} discusses the complexity and conservatism compared to existing robust (adaptive) MPC approaches, and elaborates on variations and extensions. 
%
\subsection{Polytopic tubes for mixed uncertainty}
\label{sec:main_1}
In the following, we discuss how to propagate the uncertainty of the additive disturbances $d_t\in \mathbb{D}$ and the parametric uncertainty $\Theta_t^{HC}$ using a polytopic tube and the robust MPC approach in~\cite[Prop.~1]{Robust_TAC_19}. 
We consider a standard quadratic stage cost $\ell(x,u)=\|x\|_Q^2+\|u\|_R^2$, with $Q,R$ positive definite. 
As standard in tube based robust MPC, we require the following stabilizability assumption. 
\begin{assumption}
\label{ass:stable}
Consider the prior parameter set $\Theta_0^{HC}$ from Assumption~\ref{ass:model}. 
There exist a feedback $K\in\mathbb{R}^{m\times n}$ and a positive definite matrix $P$, such that $A_{cl,\theta}:=A_{\theta}+B_{\theta}K$ is quadratically stable and satisfies
\begin{align}
\label{eq:P}
A_{cl,\theta}^\top P A_{cl,\theta}+Q+K^\top R K \preceq P,
\end{align}
for all $\theta\in\Theta_0^{HC}$. 
\end{assumption}
These matrices $P,K$ can be computed using standard robust control methods based on linear matrix inequalities (LMIs), compare e.g. Appendix~\ref{App:tube}. 
In the following, we only consider the pre-stabilized dynamics $A_{cl,\theta}$ with $u_t=Kx_t+v_t$ and the new input $v\in\mathbb{R}^m$.

The idea of tube-based robust MPC is to online predict a tube around the nominal predicted trajectory that contains all possible uncertain trajectories for different realizations of the disturbances $d\in\mathbb{D}$ and the parameters $\theta$. 
To ensure a fixed computational complexity, this tube is typically based on a fixed offline computed polytope, which is translated and scaled online, compare~\cite{kouvaritakis2016model,fleming2015robust,Robust_TAC_19,rakovic2012homothetic}.
To this end, we consider some compact polytope
\begin{align}
\label{eq:polytope}
\mathcal{P}=\{x\in\mathbb{R}^n|~H_i x\leq 1,i=1,\dots,r\},
\end{align}
the complexity of which determines the complexity of the robust MPC approach. 
Given a point estimate $\theta\in\mathbb{R}^p$, the contraction rate $\rho_{\theta}$ of the polytope $\mathcal{P}$ can be computed using the following LP
\begin{align}
\label{eq:rho}
\rho_{\theta}:=\max_i \max_{x\in\mathcal{P}} H_i A_{cl,\theta}x,
\end{align}
compare~\cite[Thm.~4.1]{blanchini1999set}.
In the considered robust MPC approach (cf.~\cite[Prop.~1]{Robust_TAC_19}), this contraction rate $\rho_{\theta}$ determines the growth of the tube (and hence the conservatism).
Standard methods to compute suitable polytopes are  the minimal robust positive invariant (RPI) set~\cite{kouramas2005minimal,rakovic2005minimal}, the maximal invariant/contractive set~\cite{pluymers2005efficient}, \cite{blanchini2008set} or maximal RPI set \cite{kerrigan2001robust},~\cite[Alg.~4.3]{fleming2016robust}.

Given the polytope $\mathcal{P}$ and feedback $K$, we compute the following constant based on an LP
\begin{align}
\label{eq:L_Theta}
L_{\mathbb{B}}:=&\max_{i,l}\max_{x\in\mathcal{P}} H_i D(x,Kx) \tilde{e}_l,
\end{align}
where $\tilde{e}_l$ denote the $2^p$ vertices of the unit hypercube $\mathbb{B}_p$. 
The following proposition shows that this constant can be used to bound the possible change in the contraction rate $\rho_{\theta}$.
\begin{proposition}
\label{prop:cont_rho}
Given parameters $\theta,\tilde{\theta}$ with $\tilde{\theta}-\theta=\Delta \theta\in\eta\mathbb{B}_p$, the following inequality holds 
\begin{align}
\label{eq:cont_rho}
\rho_{\tilde{\theta}}\leq \rho_{\theta}+\eta L_{\mathbb{B}}.
\end{align}
\end{proposition}
\begin{proof}
The property follows directly using
\begin{align*}
\rho_{\tilde{\theta}}\stackrel{\eqref{eq:rho}}{=}&\max_{i} \max_{x\in\mathcal{P}}H_i A_{cl,\tilde{\theta}}x\\
=&\max_i\max_{x\in\mathcal{P}} H_i (A_{cl,\theta}x+D(x,Kx)\Delta \theta)\\
\leq &\max_i\max_{x\in\mathcal{P}} H_i A_{cl,\theta}x+\max_i\max_{x\in\mathcal{P}} H_i D(x,Kx)\Delta \theta\\
\stackrel{\eqref{eq:rho}}{\leq} &\rho_{\theta}+\max_i\max_{ \theta_2\in\eta\mathbb{B}_p}\max_{x\in\mathcal{P}} H_ i D(x,Kx)\theta_2\\
= &\rho_{\theta}+\eta\max_{i,j}\max_{x\in\mathcal{P}} H_i D(x,Kx)\tilde{e}_j
\stackrel{\eqref{eq:L_Theta}}{=}\rho_{\theta}+\eta L_{\mathbb{B}},
\end{align*}
where the last inequality uses the fact that the function is linear in $\theta$ and $\Delta \theta\in \eta \mathbb{B}_p$.
\end{proof}
The impact of the additive disturbances is bounded with the constant $\overline{d}$, which is computed with the following LP: 
\begin{align}
\label{eq:d}
\overline{d}:=&\max_i\max_{d\in\mathbb{D}}H_i d.
\end{align}
In order to quantify the parametric uncertainty at some point $(z,v)\in\mathcal{Z}$, we define the following function in dependence of the size $\eta$ of the hypercube $\Theta^{HC}$ (cf.~\cite[Ass.~5, Prop.~2]{Robust_TAC_19}) 
\begin{align}
\label{eq:w_eta}
w_{\eta}(z,v):=&\eta\max_{i,l}H_i D(z,v)\tilde{e}_l.
\end{align}
\begin{proposition}
\label{prop:cont_w}
For any $(x,z,v)\in\mathbb{R}^{2n+m}$, $\eta\geq 0$, the function $w_{\eta}$ satisfies the following inequality 
\begin{align}
\label{eq:Lipschitz}
w_{\eta}(x,v+K(x-z))\leq w_ {\eta}(z,v)+\eta L_{\mathbb{B}}\max_i H_i (x-z).
\end{align}
\end{proposition}
\begin{proof}
First, note that the following inequality holds for any function $f:\mathbb{R}^n\rightarrow\mathbb{R}$ linear in $x$, any $x\neq 0$ and $H_i$ from the compact polytope $\mathcal{P}$ in~\eqref{eq:polytope}:
\begin{align}
\label{eq:w_eta_bound_app_help}
f(x)=f(x)\dfrac{\max_k H_k x}{\max_i H_i x}\leq \max_k H_k x \max_{\Delta x\in\mathcal{P}} f(\Delta x).
\end{align}
The claim follows using~\eqref{eq:w_eta_bound_app_help} with the linear function $H_i D(x,Kx)\tilde{e}_l$:
\begin{align*}
&w_{\eta}(x,v+K(x-z))
=\eta\max_{i,l}H_i D(x,v+K(x-z))\tilde{e}_l\\
\leq &\eta\max_{i,l} H_i D(z,v)\tilde{e}_l+\eta\max_{i,l}H_iD((x-z),K(x-z))\tilde{e}_l\\
\stackrel{\eqref{eq:w_eta_bound_app_help}}{\leq }& \tilde{w}_{\eta}(z,v)+\eta \max_k H_k(x-z) \max_{i,l}\max_{\Delta x\in\mathcal{P}} H_iD(\Delta x,K\Delta x)\tilde{e}_l\\
\stackrel{\eqref{eq:L_Theta}}{=}&\tilde{w}_{\eta}(z,v)+\eta L_{\mathbb{B}}\max_i H_i(x-z).
\end{align*}
\end{proof}
Furthermore, for each constraint~\eqref{eq:con} we compute a constant $c_j$ using the following LP:
\begin{align}
\label{eq:c_j}
c_j:=\max_{x\in\mathcal{P}}[F+GK]_j x,~j=1,\dots,q.
\end{align}

%
\subsection{Proposed robust adaptive MPC scheme}
\label{sec:main_2}
Given the constants $c_j$ $\rho$, $L_{\mathbb{B}}$ and the function $w_{\eta}$, we can state the proposed MPC scheme.
At each time $t$, given a state $x_t$, a hypercube $\Theta_t^{HC}=\overline{\theta}_t\oplus\eta_t\mathbb{B}_p$ and a LMS point estimate $\hat{\theta}_t$, we solve the following linearly constrained QP
\begin{subequations}
\label{eq:MPC}
\begin{align}
&\min_{v_{\cdot|t},w_{\cdot|t}}\sum_{k=0}^{N-1}\|\hat{x}_{k|t}\|_Q^2+\|\hat{u}_{k|t}\|_R^2+\|\hat{x}_{N|t}\|_P^2\nonumber\\
\text{s.t. }&\overline{x}_{0|t}=\hat{x}_{0|t}=x_t,~s_{0|t}=0,\\
\label{eq:dyn_nom}
&\overline{x}_{k+1|t}=A_{cl,\overline{\theta}_t}\overline{x}_{k|t}+B_{\overline{\theta}_t}v_{k|t},\\
\label{eq:dyn_hat}
&\hat{x}_{k+1|t}=A_{cl,\hat{\theta}_t}\hat{x}_{k|t}+B_{\hat{\theta}_t}v_{k|t},\\
\label{eq:s_dyn}
&s_{k+1|t}=\rho_{\overline{\theta}_t}s_{k|t}+w_{k|t},\\
\label{eq:w_scalar}
&w_{k|t}\geq \overline{d}+\eta_t(L_{\mathbb{B}}s_{k|t}+H_i D(\overline{x}_{k|t},\overline{u}_{k|t})\tilde{e}_l),\\
\label{eq:con_tightend}
&F_j\overline{x}_{k|t}+G_j\overline{u}_{k|t}+c_js_{k|t}\leq 1,\\
\label{eq:u}
&\overline{u}_{k|t}=v_{k|t}+K\overline{x}_{k|t},\quad 
\hat{u}_{k|t}=v_{k|t}+K\hat{x}_{k|t},\\
\label{eq:con_term}
&(\overline{x}_{N|t},s_{N|t})\in\mathcal{X}_f,\\
&j=1,\dots,q,~ k=0,\dots,N-1,\nonumber\\
&~i=1,\dots r,~ l=1,\dots,2^p,\nonumber
\end{align}
\end{subequations}
where $\mathcal{X}_f$ is a terminal set to be specified later. 
The solution to~\eqref{eq:MPC} is denoted by $\overline{x}^*_{\cdot|t}$, $\hat{x}^*_{\cdot|t}$, $v^*_{\cdot|t}$, $\hat{u}^*_{\cdot|t}$, $\overline{u}^*_{\cdot|t}$, $w^*_{\cdot|t}$, $s^*_{\cdot|t}$. 
The proposed scheme ensures robust constraint satisfaction by predicting a polytopic tube $\mathbb{X}_{k|t}=\{z|~H_i(z-\overline{x}_{k|t})\leq s_{k|t}\}$ that contains all possible future trajectories of the uncertain system~\eqref{eq:sys} subject to the input trajectory $v_{\cdot|t}$. 
The uncertainty propagation is achieved with  the scalar tube dynamics of $s$~\eqref{eq:s_dyn}, and the construction of $w$ in~\eqref{eq:w_scalar}, which overapproximates the uncertainty of all states and inputs within the tube $\mathbb{X}_{k|t}$, compare~\cite[Prop.~2]{Robust_TAC_19}. 

Compared to the robust MPC scheme in~\cite{Robust_TAC_19}, the tube propagation depends on the online set estimate $\Theta_t^{HC}$, thus reducing the conservatism. 
Furthermore, similar to~\cite{aswani2013provably,di2016indirect,lorenzen2019robust} a second predicted trajectory $\hat{x}_{k|t}\in\mathbb{X}_{k|t}$ based on the LMS point estimate $\hat{\theta}_t$ (c.f.~\eqref{eq:dyn_hat}) is used to evaluate the predicted cost function, which improves performance. 

Compared to a nominal MPC scheme, the evaluation of the uncertainty $w_{\eta}(x,u)$~\eqref{eq:w_eta} along the prediction horizon $N$  introduces $N\cdot r\cdot 2^p$ additional linear inequality constraints~\eqref{eq:w_scalar}.  
The computational complexity and conservatism in comparison to other robust MPC methods is discussed in Remark~\ref{rk:robust_methods} and investigated in the numerical example in Section~\ref{sec:num}.

The overall offline and online computations are summarized
in Algorithm~\ref{alg:adaptive_offline} and \ref{alg:adaptive_online}, respectively.
\begin{algorithm}[h]
\caption{Robust adaptive MPC - Offline}
\label{alg:adaptive_offline}
\begin{algorithmic}
\State Given model, parameter set $\Theta_0^{HC}$ (Ass.~\ref{ass:model}), constraints~\eqref{eq:con}.
\State Compute feedback $K$, terminal cost $P$ (Ass.~\ref{ass:stable}). 
\State Set parameter update gain $\mu>0$ (Lemma~\ref{lemma:LMS}).
\State Design polytope $\mathcal{P}$~\eqref{eq:polytope}. 
\State Compute $\rho_{\overline{\theta}_0}$, $L_{\mathbb{B}}$, $\overline{d}$, $c_j$ using LPs \eqref{eq:rho}, \eqref{eq:L_Theta}, \eqref{eq:d}, \eqref{eq:c_j}.
\State Check if condition~\eqref{eq:term_cond} in Prop.~\ref{prop:term} holds. 
\end{algorithmic}
\end{algorithm}
\begin{algorithm}[h]
\caption{Robust adaptive MPC - Online}
\label{alg:adaptive_online}
Execute at each time step $t\in\mathbb{N}$:
\begin{algorithmic}
\State Measure state $x_t$.
\State Update $\Theta_t^{HC}$ using Algorithm~\ref{alg:HC}. 
\State Update $\hat{\theta}_t$ using~\eqref{eq:theta_LMS}. 
\State Update $\rho_{\overline{\theta}_t}$ using~\eqref{eq:rho}. 
\State Solve MPC optimization problem~\eqref{eq:MPC}.
\State Apply control input $u_t=v^*(0|t)+Kx_t$.
\end{algorithmic}
\end{algorithm}

%
\subsection{Theoretical analysis}
\label{sec:main_3}
In the following, we detail the theoretical analysis and provide the technical conditions on the terminal set.

\subsubsection*{Terminal set}
The following assumption captures the desired properties of the terminal ingredients. 
\begin{assumption}\label{ass:term}
Consider the set $\Theta_0$ and matrix $K$ from Assumptions~\ref{ass:model} and \ref{ass:stable}. 
There exists a terminal region $\mathcal{X}_{f}\subset\mathbb{R}^{n+1}$ such that the following properties hold 
\begin{description}
\item[] for all $\Theta^{HC}=\bar{\theta}\oplus\eta\mathbb{B}_p\subseteq\Theta_0^{HC},$
\item[] for all $(x,s)\in\mathcal{X}_{f}$,
\item[] for all $\tilde{s}\in\mathbb{R}_{\geq0}$,
\item[] for all $s^+\in [0,(\rho_{\overline{\theta}}+\eta L_{\mathbb{B}})s+w_{\eta}(x,Kx)+\overline{d}-\tilde{s}]$
\item[] for all $x^+\;\textrm{s.t.}\;\max_{i}H_i(x^+-A_{cl,\overline{\theta}}x)\leq \tilde{s}$:
\end{description}
\begin{subequations}
\begin{align}
\label{eq:ass_term_2}
	&(x^+,s^+)\in\mathcal{X}_{f},\\
\label{eq:ass_term_3}
	&[F+GK]_jx+c_js\leq 1, \quad j=1,...,q.	
\end{align}
\end{subequations}
\end{assumption}
Condition~\eqref{eq:ass_term_2} ensures robust positive invariance of the terminal region, and condition~\eqref{eq:ass_term_3} ensures that the tightened state and input constraints are satisfied within the terminal region. 
The following proposition provides a simple polytopic terminal set constraint $\mathcal{X}_f$.
\begin{proposition}
\label{prop:term}
Let Assumption~\ref{ass:model} and \ref{ass:stable} hold.
Suppose that the following condition holds 
\begin{align}
\label{eq:term_cond}
\rho_{\overline{\theta}_0}+\eta_0L_{\mathbb{B}}+c_{\max}\overline{d}\leq 1,
\end{align}
with $c_{\max}=\max_jc_j$. 
Then the polytopic terminal set 
\begin{align}
\label{eq:term_set}
\mathcal{X}_f=\{(x,s)\in\mathbb{R}^{n+1}|~c_{\max}(s+H_ix)\leq 1,~i=1,\dots,r\},
\end{align}
satisfies Assumption~\ref{ass:term}. 
\end{proposition}
\begin{proof}
Note that $\Theta^{HC}\subseteq\Theta_0^{HC}$ (Lemma~\ref{lemma:HC}) implies $\overline{\theta}\subseteq\overline{\theta}_0\oplus(\eta_0-\eta)\mathbb{B}_p$. 
Thus, satisfaction of~\eqref{eq:term_cond} implies satisfaction of~\eqref{eq:term_cond} with $\eta_0,\overline{\theta}_0$ replaced by $\eta,\overline{\theta}$, by using Proposition~\ref{prop:cont_rho} and
\begin{align*}
	\rho_{\overline{\theta}}+\eta L_{\mathbb{B}}
\stackrel{\eqref{eq:cont_rho}}{\leq}\rho_{\overline{\theta}_0}+(\eta_0-\eta)L_{\mathbb{B}}+\eta L_{\mathbb{B}}=\rho_{\overline{\theta}_{0}}+\eta_{0}L_{\mathbb{B}}.
\end{align*}
Satisfaction of the tightened constraints~\eqref{eq:ass_term_3} follows from the definition of $c_j$ in~\eqref{eq:c_j}, using
\begin{equation}
\begin{aligned}
	[F+GK]_jx+c_js 
\stackrel{\eqref{eq:c_j}}{\leq}
	&c_j \max_i H_ix+c_js\\
        \leq &c_{\max}(\max_i H_i x+s)\stackrel{\eqref{eq:term_set}}{\leq} 1.
\end{aligned}
\end{equation}
The uncertainty in the terminal set can be bounded as follows by using Prop.~\ref{prop:cont_w}:
\begin{equation}
\begin{aligned}\label{eq_bound_w_tilde_terminal}
&	\max_{(x,s)\in\mathcal{X}_{f}}w_{\eta}(x,Kx)+\eta L_{\mathbb{B}}s\\
\stackrel{\eqref{eq:Lipschitz}}{\leq}&w_{\eta}(0,0)+\eta L_{\mathbb{B}}\max_{(x,s)\in\mathcal{X}_f}(\max_i H_i (x-0)+s)\\
\stackrel{\eqref{eq:term_set}}{\leq} &\eta L_{\mathbb{B}}/c_{\max}.
\end{aligned}
\end{equation}
The state $x^+$ can be bounded as
\begin{align}
\label{eq:bound_terminal_x_plus}
	\max_{i}H_ix^+&\leq\max_{i}H_iA_{cl,\overline{\theta}}x+\tilde{s}\stackrel{\eqref{eq:rho}}{\leq}\rho_{\bar{\theta}}\max_{i}H_ix+\tilde{s}.
\end{align}
The robust positive invariance condition~\eqref{eq:ass_term_2} follows from
\begin{align*}
s^+ + H_i x^+\leq&\rho_{\overline{\theta}}s+\eta L_{\mathbb{B}}s+w_{\eta}(x,Kx)+\overline{d}-\tilde{s}\\
&+ H_i A_{cl,\overline{\theta}}x+\tilde{s}\\
\stackrel{\eqref{eq_bound_w_tilde_terminal}\eqref{eq:bound_terminal_x_plus}}{\leq}& \rho_{\overline{\theta}}(s+\max_i H_i x)+\eta L_{\mathbb{B}}/c_{\max}+\overline{d}\\
\stackrel{\eqref{eq:term_set}}{\leq}& (\rho_{\overline{\theta}}+\eta L_{\mathbb{B}})/c_{\max}+\overline{d}\stackrel{\eqref{eq:term_cond}}{\leq } 1/c_{\max}.
\end{align*}
\end{proof}
This proposition provides a simple and intuitive characterization of the terminal set $\mathcal{X}_f$, if condition~\eqref{eq:term_cond} is satisfied.   
Generalizations of this design for the stabilization of steady-state $x_s\neq 0$ are discussed in Appendix~\ref{App:terminal}. 
\subsubsection*{Closed loop properties}
The following theorem is the main result of this paper and establishes the closed-loop properties
of the proposed robust adaptive MPC scheme.
\begin{theorem}
\label{thm:main}
Let Assumptions~\ref{ass:model}, \ref{ass:stable}, and \ref{ass:term} hold.
Suppose that Problem~\eqref{eq:MPC} is feasible at $t=0$. 
Then~\eqref{eq:MPC} is recursively feasible and the constraints~\eqref{eq:con} are satisfied for the resulting closed-loop system. 
Furthermore, if the update gain $\mu>0$ satisfies $1/\mu>\sup_{(x,u)\in\mathcal{Z}}\|D(x,u)\|^2$, 
the closed loop is finite gain $\mathcal{L}_2$ stable, i.e., there exist constants $c_0,c_1,c_2>0$, such that the following inequality holds for all $T\in\mathbb{N}$
\begin{align*}
\sum_{k=0}^{T}\|x_k\|^2\leq c_0\|x_0\|^2+c_1\|\hat{\theta}_0-\theta^*\|^2+c_2\sum_{k=0}^{T}\|d_k\|^2.
\end{align*}
\end{theorem}
\begin{proof}
The proof of robust recursive feasibility is an extension of~\cite[Thm. 1]{Robust_TAC_19} to online adapting models. 
 The stability result follows using the same arguments as in~\cite[Thm.~14]{ lorenzen2019robust}. 
The main novel step to show recursive feasibility is to prove that the tube around the candidate solution is contained inside the tube of the optimal solution at time $t$, compare Figure~\ref{fig:illustrate} for an illustration.  \\
As done in \cite[Thm. 1]{Robust_TAC_19} we first construct the candidate solution (Part I). 
Then we bound the difference between this candidate and the optimal solution at time $t$ (Part II) in order to ensure that the new tube around the candidate trajectory is contained in the previous optimal tube (Part III).
Then we show that the candidate solution satisfies the tightened state and input constraints (Part IV) and the terminal set constraint (Part V).
Finally, we establish the stability properties (Part VI). \\
\textbf{Part I. } Candidate solution:
For convenience, define
\begin{align}
\label{eq:solution_old}
	v^*_{N|t}&=0,~\overline{u}^*_{N|t}=K\overline{x}^*_{N|t},\\
	w^*_{N|t}&=\overline{d}+
\eta_tL_{\mathbb{B}}s^*_{N|t}+w_{\eta_t}(\overline{x}^*_{N|t},\overline{u}^*_{N|t}),\nonumber\\
\overline{x}^*_{N+1|t}&=A_{cl,\overline{\theta}_t}\overline{x}^*_{N|t}+B_{\overline{\theta}_t}v^*_{N|t},~s^*_{N+1|t}=\rho_{\overline{\theta}_t}s^*_{N|t}+w^*_{N|t}. \nonumber
\end{align}
We consider the candidate solution
\begin{align}
\label{eq:candidate}
	\overline{x}_{0|t+1}&=\hat{x}_{0|t+1}=x_{t+1},\quad s_{0|t+1}=0, \\
	v_{k|t+1}&=v^*_{k+1|t},\nonumber\\
        \overline{u}_{k|t+1}&=v_{k|t+1}+K\overline{x}_{k|t+1},~
        \hat{u}_{k|t+1}=v_{k|t+1}+K\hat{x}_{k|t+1},\nonumber\\
	\overline{x}_{k+1|t+1}&=A_{cl,\overline{\theta}_{t+1}}\overline{x}_{k|t+1}+B_{\overline{\theta}_{t+1}}v_{k|t+1},\nonumber	\\
	\hat{x}_{k+1|t+1}&=A_{cl,\overline{\theta}_{t+1}}\hat{x}_{k|t+1}+B_{\overline{\theta}_{t+1}}v_{k|t+1},\nonumber	\\
	s_{k+1|t+1}&=\rho_{\overline{\theta}_{t+1}}s_{k|t+1}+w_{k|t+1},\nonumber\\
	w_{k|t+1}&=\overline{d}+\eta_{t+1}L_{\mathbb{B}}s_{k|t+1}+w_{\eta_{t+1}}(\overline{x}_{k|t+1},\overline{u}_{k|t+1}),\nonumber
\end{align}
with $k=0,...,N-1$.\\
\textbf{Part II. } Bound candidate solution: 
Due to the parameter update, the candidate state trajectory $\overline{x}_{\cdot|t+1}$ is computed with a different model than the previous optimal solution $\overline{x}^*_{\cdot|t}$. 
The dynamics of the optimal trajectory $\overline{x}^*$ can be equivalently written as
\begin{align}
\label{eq:dynamic_parameter_update}
	&\overline{x}^*_{k+2|t}=A_{cl,\overline{\theta}_t}\overline{x}^*_{k+1|t}+B_{\overline{\theta}_t}v^*_{k+1|t}\\
     =&A_{cl,\overline{\theta}_{t+1}}\overline{x}^*_{k+1|t}+B_{\overline{\theta}_{t+1}}v^*_{k+1|t}
       -D(\overline{x}^*_{k+1|t},\overline{u}^*_{k+1|t})\Delta\overline{\theta}_t,\nonumber
\end{align}
with the change in parameters $\Delta\overline{\theta}_t=\overline{\theta}_{t+1}-\overline{\theta}_t$. 
Note that the definition of $\overline{\theta}_{t+1}$ in Algorithm~\ref{alg:HC} ensures  $\Delta\overline{\theta}_t\in\Delta\Theta_t:=\Delta \eta_{t}\mathbb{B}_p$, with $\Delta\eta_t=\eta_{t}-\eta_{t+1}\geq 0$. 
Define the auxiliary tube size
\begin{align}
\label{eq:s_tilde}
\tilde{s}_{0|t+1}=&w^*_{0|t}=s^*_{1|t},\nonumber\\
\tilde{s}_{k+1|t+1}=&\rho_{\overline{\theta}_{t+1}}\tilde{s}_{k|t+1}+w_{\Delta \eta}(\overline{x}^*_{k+1|t},\overline{u}^*_{k+1|t}).
\end{align}
In the following, we show that that the candidate solution $\overline{x}_{\cdot|t+1}$, the previous optimal solution $\overline{x}^*_{\cdot|t+1}$ and the auxiliary tube size $\tilde{s}_{\cdot|t+1}$ satisfy 
\begin{align}
\label{eq:tilde_s_bound}
\overline{x}_{k|t+1}-\overline{x}^*_{k+1|t}=:e_{k|t+1}\in \tilde{s}_{k|t+1}\cdot \mathcal{P},
\end{align}
 where  $e_{\cdot|t+1}$ denotes the error between the two trajectories. 
Using~\eqref{eq:candidate} and \eqref{eq:dynamic_parameter_update}, the error dynamics are given by
\begin{align*}
e_{k+1|t+1}=&A_{cl,\overline{\theta}_{t+1}}e_{k|t+1}
+D(\overline{x}^*_{k+1|t},\overline{u}^*_{k+1|t})\Delta\overline{\theta}_{t},
\end{align*}
with the initial condition $e_{0|t+1}=x_{t+1}-x^*_{1|t}$. 
We prove~\eqref{eq:tilde_s_bound} using induction. 
Induction start $k=0$:
Using $\theta^*\in\overline{\theta}_t\oplus\eta_t\mathbb{B}_p=\Theta_t^{HC}$, condition~\eqref{eq:tilde_s_bound} is satisfied at $k=0$ with
\begin{align*}
 &H_i(\overline{x}_{0|t+1}-\overline{x}^*_{1|t})=H_i(x_{t+1}-\overline{x}^*_{1|t})\\
=&H_i(d_t+D(\overline{x}^*_{0|t},\overline{u}^*_{0|t})(\theta^*-\overline{\theta}_t))\\
\leq&\max_i H_id_t+\eta_t\max_{i,l}  H_i D(\overline{x}^*_{0|t},\overline{u}^*_{0|t})\tilde{e}_l\\
\stackrel{\eqref{eq:d},\eqref{eq:w_eta}}{\leq}& \overline{d}+w_{\eta_t}(\overline{x}^*_{0|t},\overline{u}^*_{0|t}) 
\stackrel{\eqref{eq:w_scalar}}{\leq} w^*_{0|t}\stackrel{\eqref{eq:s_tilde}}{=}\tilde{s}_{0|t+1}.
\end{align*}
Induction step $k+1$:
 Assuming~\eqref{eq:tilde_s_bound} holds at some $k\in\{0,\dots,N-1\}$, condition~\eqref{eq:tilde_s_bound} also holds at $k+1$ using
\begin{align*}
&\max_{i}H_ie_{k+1|t+1}\\
\leq& 
\max_{i}H_iA_{cl,\overline{\theta}_{t+1}}e_{k|t+1}+\max_{i}H_iD(\overline{x}^*_{k+1|t},\overline{u}^*_{k+1|t})\Delta \overline{\theta}_t\\
\stackrel{\eqref{eq:rho}}{\leq}&\rho_{\overline{\theta}_{t+1}}\max_i H_i e_{k|t+1}+w_{\Delta \eta}(\overline{x}^*_{k+1|t},\overline{u}^*_{k+1|t})\\
\stackrel{\eqref{eq:tilde_s_bound}}{\leq} & \rho_{\overline{\theta}_{t+1}}\tilde{s}_{k|t+1}+w_{\Delta \eta}(\overline{x}^*_{k+1|t},\overline{u}^*_{k+1|t})\\
\stackrel{\eqref{eq:s_tilde}}{=}&\tilde{s}_{k+1|t+1},
\end{align*}
where the second inequality uses $\Delta\overline{\theta}_t\subseteq\Delta \eta_t\mathbb{B}_p$. 
\subsubsection*{Interpretation}
The trajectory $\tilde{s}$ is composed of two parts: a first part depending on the initial prediction mismatch $x_{t+1}-\overline{x}^*_{1|t}$ and a second part due to the parameter update. 
In the special case that there is no model adaptation ($\Delta\eta_t=0,~\Delta\overline{\theta}=0$), the auxiliary tube size $\tilde{s}$ reduces to $\tilde{s}_{k|t+1}=\rho^kw^*_{0|t}$, which shows that the robust MPC proof in~\cite{Robust_TAC_19} is contained as a special case. \\
\textbf{Part III. } Prove that the tube around the candidate solution at time $t+1$ is contained inside the tube around the optimal solution at time $t$. 
We show that the following inequality holds for $k=0,\dots,N$, using induction:
\begin{align}
 \label{eq:s_bound_induction}
	&s_{k|t+1}- s^*_{k+1|t}+\tilde{s}_{k|t+1}\leq 0.
\end{align}
Inequality~\eqref{eq:s_bound_induction} ensures that the candidate tube $\mathbb{X}_{k|t+1}$ is contained in the previous optimal tube $\mathbb{X}^*_{k+1|t}$. 
This nestedness property of the tubes is illustrated in Figure~\ref{fig:illustrate}.
The key to showing this property is that the reduction in the uncertainty $w_{k|t+1}-w^*_{k+1|t}$ is equivalent to the uncertainty that is used to compute the tube $\tilde{s}$~\eqref{eq:s_tilde} that bounds the deviation between the previous optimal solution and the candidate solution. 
First, note that the following bound holds for all $k=0,\dots,N-1$ using $\Delta\eta_t\geq 0$ and the bound~\eqref{eq:tilde_s_bound}:
\begin{align}
\label{eq:w_bound_induction_2}
&w_{k|t+1}-w^*_{k+1|t}+\Delta \eta_t L_{\mathbb{B}} s^*_{k+1|t}\\
\stackrel{\eqref{eq:w_eta}\eqref{eq:w_scalar}\eqref{eq:candidate}}{\leq}& L_{\mathbb{B}}(\eta_{t+1}s_{k|t+1}-\eta_ts^*_{k+1|t})+\Delta \eta_t L_{\mathbb{B}} s^*_{k+1|t}\nonumber\\
&+w_{\eta_{t+1}}(\overline{x}_{k|t+1},\overline{u}_{k|t+1})-w_{\eta_t}(\overline{x}^*_{k+1|t},\overline{u}^*_{k+1|t})\nonumber\\
\stackrel{\eqref{eq:Lipschitz}}{\leq}& \eta_{t+1}L_{\mathbb{B}}(s_{k|t+1}-s^*_{k+1|t})\nonumber\\
&+\eta_{t+1}L_{\mathbb{B}}\max_i H_i(\overline{x}_{k|t+1}-\overline{x}^*_{k+1|t})\nonumber\\
&+w_{\eta_{t+1}}(\overline{x}^*_{k+1|t},\overline{u}^*_{k+1|t})-w_{\eta_t}(\overline{x}^*_{k+1|t},\overline{u}^*_{k+1|t}).\nonumber\\
\stackrel{\eqref{eq:tilde_s_bound}}{\leq}& \eta_{t+1}L_{\mathbb{B}}(s_{k|t+1}-s^*_{k+1|t}+\tilde{s}_{k|t+1})\nonumber\\
&-w_{\Delta\eta_t}(\overline{x}^*_{k+1|t},\overline{u}^*_{k+1|t}).\nonumber
\end{align}
 Induction start $k=0$: Condition~\eqref{eq:s_bound_induction} holds with 
\begin{align*}
&s_{0|t+1}-s^*_{1|t}+\tilde{s}_{0|t+1}=w^*_{0|t}-w^*_{0|t}=0.
\end{align*}
Suppose that condition~\eqref{eq:s_bound_induction} holds for some $k\in\{0,\dots,N-1\}$, then condition~\eqref{eq:s_bound_induction} also holds at $k+1$ with
\begin{align*}
&s_{k+1|t+1}+\tilde{s}_{k+1|t+1}-s^*_{k+2|t}\\
\stackrel{\eqref{eq:s_dyn}\eqref{eq:candidate}\eqref{eq:s_tilde}}{=}&\rho_{\overline{\theta}_{t+1}}s_{k|t+1}+w_{k|t+1}
+\rho_{\overline{\theta}_{t+1}}\tilde{s}_{k|t+1}\\
&+w_{\Delta\eta_t}(\overline{x}^*_{k+1|t},\overline{u}^*_{k+1|t})-\rho_{\overline{\theta}_t}s^*_{k+1|t}-w^*_{k+1|t}\\
\stackrel{\eqref{eq:w_bound_induction_2}}{\leq}&\rho_{\overline{\theta}_{t+1}}(s_{k|t+1}+\tilde{s}_{k|t+1})				
-(\rho_{\overline{\theta}_t}+\Delta\eta_tL_{\mathbb{B}})s^*_{k+1|t}\\
&\eta_{t+1}L_{\mathbb{B}}(s_{k|t+1}-s^*_{k+1|t}+\tilde{s}_{k|t+1})\\
\stackrel{\eqref{eq:cont_rho}}{\leq} &(\rho_{\overline{\theta}_t}+\Delta \eta_tL_{\mathbb{B}}+\eta_{t+1}L_{\mathbb{B}})\\
&\cdot(s_{k|t+1}+\tilde{s}_{k|t+1}-s^*_{k+1|t})
\stackrel{\eqref{eq:s_bound_induction}}{\leq} 0.
\end{align*}
\textbf{Part IV.} State and input constraint satisfaction \eqref{eq:con_tightend}.
For $k=0,...,N-2$ we have
\begin{equation*}
\begin{aligned}
	&F_j\overline{x}_{k|t+1}+G_j\overline{u}_{k|t+1}+c_js_{k|t+1}\\
	\stackrel{(\ref{eq:c_j})}{\leq}
	&F_j\overline{x}^*_{k+1|t}+G_j\overline{u}^*_{k+1|t}\\
        &+c_j\max_{i}H_i(\overline{x}_{k|t+1}-\overline{x}^*_{k+1|t})+c_js_{k|t+1}			\\
	\stackrel{\eqref{eq:tilde_s_bound}}{\leq}
	& F_j\overline{x}^*_{k+1|t}+G_j\overline{u}^*_{k+1|t}+c_j(\tilde{s}_{k|t+1}+s_{k|t+1})	\\
	\stackrel{(\ref{eq:s_bound_induction})}{\leq}
	& F_j\overline{x}^*_{k+1|t}+G_j\overline{u}^*_{k+1|t}+c_js^*_{k+1|t} 
	\stackrel{\eqref{eq:con_tightend}}{\leq} 1.
\end{aligned}
\end{equation*}
For $k=N-1$ the terminal ingredients (Ass.~\ref{ass:term}) ensure
\begin{equation*}
\begin{aligned}
	&F_j\overline{x}_{N-1|t+1}+G_j\overline{u}_{N-1|t+1}+c_js_{N-1|t+1}						\\
	\stackrel{\eqref{eq:c_j}\eqref{eq:tilde_s_bound}\eqref{eq:s_bound_induction}}{\leq}&
	F_j\overline{x}^*_{N|t}+G_j\overline{u}^*_{N|t}+c_js^*_{N|t}							\\
	\stackrel{\eqref{eq:solution_old}}{=}&
	[F+GK]_j\overline{x}^*_{N|t}+c_js^*_{N|t}
	 \stackrel{\eqref{eq:con_term}\eqref{eq:ass_term_3}}{\leq} 1.
\end{aligned}
\end{equation*}
Satisfaction of~\eqref{eq:con_tightend} at $k=0$ ensures that the closed loop satisfies the constraints~\eqref{eq:con}, i.e.,  $(x_t,u_t)\in\mathcal{Z}$ for all $t\geq 0$. \\
\textbf{Part V. }Terminal constraint satisfaction.
We have
\begin{align*}
	\max_{i}H_i(\overline{x}_{N|t+1}-\overline{x}^*_{N+1|t})\stackrel{\eqref{eq:tilde_s_bound}}{\leq}\tilde{s}_{N|t+1}
\end{align*}	
and
\begin{align*}
	s_{N|t+1}\stackrel{\eqref{eq:s_bound_induction}}{\leq}&s^*_{N+1|t}-\tilde{s}_{N|t+1}\\
\stackrel{\eqref{eq:solution_old}}{=}&(\rho_{\overline{\theta}_t}+\eta_tL_{\mathbb{B}})s^*_{N|t}     +\overline{d}+w_{\eta_t}(\overline{x}^*_{N|t},\overline{u}^*_{N|t})-\tilde{s}_{N|t+1}.
\end{align*}
Thus, condition \eqref{eq:ass_term_2} from Assumption~\ref{ass:term} ensures $(\overline{x}_{N|t+1},s_{N|t+1})\in\mathcal{X}_{f}$. \\
\textbf{Part VI. } Finite-gain $\mathcal{L}_2$ stability: 
The finite-gain $\mathcal{L}_2$ stability of the closed loop follows from the properties of the LMS point estimate (Lemma~\ref{lemma:LMS}), constraint satisfaction $(x_t,u_t)\in\mathcal{Z}$ and quadratic bounds on the value functions~\cite[Thm.~14]{ lorenzen2019robust}, compare also~\cite[Prop.~5.10]{Elisa}.  
\end{proof}
\begin{figure}[tbp]
\begin{center}
\includegraphics[width=0.425\textwidth]{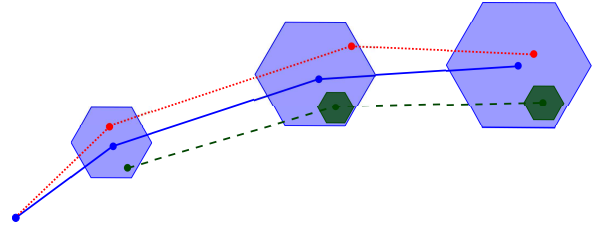}
\end{center}
\caption{Illustration - nested tubes property: Optimal trajectory $x^*_{\cdot|t}$ (blue, solid), candidate trajectory $x_{\cdot|t+1}$ (green, dashed), LMS trajectory $\hat{x}^*_{\cdot|t}$ (red, dotted), with correspond tubes  $\mathbb{X}^*_{k|t}=\{z|H_i (z-\overline{x}^*_{k|t})\leq s^*_{k|t}\}$ (blue polytopes), $\mathbb{X}_{k|t+1}=\{\tilde{z}|~H_i(\tilde{z}-\overline{x}_{k|t+1})\leq s_{k|t+1})\}$ (green polytopes).}
\label{fig:illustrate}
\end{figure}

%
\subsection{Discussion}
\label{sec:discuss}
\begin{remark}
\label{rk:time-varying} (Time-varying parameters)
The proposed robust adaptive MPC scheme can be directly extended to account for slowly time-varying parameters $\theta_{t+1}\in\Theta\cap (\theta_t\oplus\Omega)$, where the hypercube $\Omega=\omega\mathbb{B}_p$ bounds the maximal change in the time-varying parameters $\theta_t$.
In this case, Algorithm~\ref{alg:HC} uses the non-falsified set $\Delta_{k|t}=\Delta_k\oplus (t-k)\Omega$ instead of $\Delta_k$. 
Similarly, for the predictions a growing parameter set $\Theta_{k|t}^{HC}=\Theta_t^{HC}\oplus k\Omega$ is considered with $\eta_{k|t}=\eta_t+k\omega$. 
The finite-gain stability w.r.t the additive disturbances $d_t$ from Theorem~\ref{thm:main}  changes to a finite-gain stability w.r.t both the disturbances $d_t$ and a signal measuring the change in the parameters $\theta_t$. 
The other theoretical properties in Theorem~\ref{thm:main} remain unchanged.  
The corresponding details and proofs can be found in~{\cite[Sec.~5.2]{Elisa}}. 
\end{remark}

\begin{remark}
\label{rk:robust_methods}
(Complexity and conservatism of robust MPC methods)
In the following, we discuss different methods to propagate the uncertainty in robust MPC in terms of their computational complexity and conservatism. 
In the considered robust MPC approach, the uncertainty is propagated using 
\begin{align}
\label{eq:w_2}
s^+\geq (\rho_{\theta}+\eta L_{\mathbb{B}})s+\overline{d}+\eta H_i D(\overline{x},\overline{u})\tilde{e}_l,
\end{align}
with $r\cdot 2^p$ linear inequality constraints and the additional optimization variable $s$.
This formulation takes into account the parametric uncertainty and the shape of the polytope $H_i$. 
The contractive dynamics and the disturbances are overapproximated with the scalars $\overline{d}$, $\rho_{\theta}$.

In case only $p_B<p$ uncertain parameters $\theta$ affect the matrix $B_{\theta}$, the following formulation provides a computationally cheaper overapproximation using inequality~\eqref{eq:Lipschitz} with $(z,v)=(0,u-Kx)$: 
\begin{align}
\label{eq:w_3}
s^+\geq(\rho_{\theta}+\eta L_{\mathbb{B}})s+\eta L_{\mathbb{B}}H_k\overline{x}+\eta H_iD(0,v)\tilde{e}_l+\overline{d}.
\end{align}
By introducing an additional slack variable to evaluate $\max_k H_k x$, this condition can be posed as $r\cdot (2^{p_B}+1)$ inequality constraints, compare~Appendix~\ref{App:tube}. 

In~ \cite[Prop.~1]{Robust_TAC_19}, the following formulation was considered
\begin{align}
\label{eq:w_1}
s^+\geq (\rho_{\theta}+\eta L_{\mathbb{B}})s+\eta H_iD(\overline{x},\overline{u})\tilde{e}_l+\overline{d}_i.
\end{align}
with $\overline{d}_i=\max_{d\in\mathbb{D}}H_i d$. 
This constraint is less conservative than \eqref{eq:w_2} since $\overline{d}\geq \overline{d}_i$, but the presented proof for robust \textit{adaptive} MPC cannot be applied to this formulation since $\max_i \eta H_iD(\overline{x},\overline{u})\tilde{e}_l+\overline{d}_i$ is piece-wise affine in $\eta$, while $w_{\eta}+\overline{d}$ in~\eqref{eq:w_2} is affine in $\eta$.

In~\cite{Lu2019RAMPC} a robust adaptive MPC scheme is presented with the more flexible parameterization $\mathbb{X}=\{z|~H_iz\leq \alpha_i,~i=1,\dots,r\}$, using an online optimized vector $\alpha\in\mathbb{R}^r$, compare also~\cite{fleming2015robust}.  
The corresponding tube propagation is given by 
\begin{align}
\label{eq:Lu_19}
\alpha_i^+\geq \overline{d}_i+[1,(\overline{\theta}+\eta\tilde{e}_l)^\top]\hat{H}_i\alpha+H_iB(\overline{\theta}+\eta\tilde{e}_l)v,
\end{align}
where $\hat{H}_i$ are computed offline using LPs, compare~\cite[Lemma~8]{Lu2019RAMPC} for details.
This conditions uses $r\cdot 2^p$ linear inequality constraints and requires $r$ additional optimization variables $\alpha$.
Compared to~\eqref{eq:w_2}, the matrices $\hat{H}_i$ capture the propagation more accurately compared to the scalars $\rho,~L_{\mathbb{B}}$. 
Furthermore, the more flexible parameterization reduces the conservatism at the expense of more decision variables.

In~\cite{lorenzen2019robust}, a homothetic tube~\cite{rakovic2012homothetic}, \cite{rakovic2013homothetic} is used, where online optimized matrices $\Lambda^j_{k|t}$  characterize the tube propagation as follows
\begin{align}
\label{eq:homothetic_tube}
&H_i(D(\overline{x},\overline{u})+s D(z^j,Kz^j))=\Lambda_i^j H_{\theta},~\Lambda_i^j\geq 0,\\
&\Lambda_i^j h_{\theta} + H_i (A_{cl,0}\overline{x}+B_0v-x^++s A_{cl,0}z^j)+\overline{d}_i\leq s^+,\nonumber
\end{align}
with $\Theta=\{H_{\theta}\theta\leq h_{\theta}\}$ and the vertex representation $\mathcal{P}=Conv(z^j)$, $j=1,\dots,r_{v}$. 
Assuming $\Theta$ is a hypercube and thus $\Lambda^j\in\mathbb{R}^{r\times 2p}$, this condition requires  $r_{v}\cdot r\cdot 2p$ additional optimization variables and $3r_{v}\cdot r$ additional inequality constraints. 
Since the complexity increases with $r_{v}\cdot r$, the method is likely limited to low dimensional problems with simple polytopes $\mathcal{P}$. 

This trade-off between computational complexity and conservatism is also investigated in the numerical example in Section~\ref{sec:num}. 

{We point out that the results in Theorem~\ref{thm:main} only apply to formulation~\eqref{eq:w_2}. 
Extending the results in Theorem~\ref{thm:main} to a more general class of robust tube formulations (similar to~\cite{Robust_TAC_19}) is part of current research.}
\end{remark}

\begin{remark}
\label{rk:fixed_shape}
In this paper, we consider a hypercube $\Theta_t^{HC}$ for the uncertainty propagation to keep the computational complexity low. 
Instead, any polytope $\Theta_t$  of fixed shape can equally be used with $\Theta_t=\eta_t \Theta_0$. 
However, the theoretical properties in Theorem~\ref{thm:main} are not valid if the shape of the set $\Theta_t$ changes. 
In particular, this means that we cannot use a general hyperbox $\Theta_t^{HB}$, with a flexible ratio between the different side lengths, which is the case in~\cite{lorenzen2019robust,Lu2019RAMPC}. 
Since using a hyperbox instead of a hypercube could reduce the conservatism without increasing the computational complexity, extending the theory to allow for such sets is an interesting open problem. 
Similarly, in~\cite{adetola2011robust} for nonlinear systems a  fixed shape set $\Theta_t$ in form of a ball is used for the robust propagation, even though a less conservative ellipsoidal set  $\Theta=\{\|\tilde{\theta}\|_\Sigma^2\leq \eta\}$ based on the RLS estimate $\hat{\theta}$ is available, compare also~\cite{zhang2019computationally}. 
\end{remark}

\section{Numerical example}
\label{sec:num}
The following example demonstrates the performance improvements of the proposed adaptive method compared to a robust MPC formulation.
Furthermore, we investigate the computational demand and the conservatism of the different robust MPC formulations. 
\subsubsection*{Model}
We consider a simple mass spring damper system 
\begin{align*}
m\ddot{x}_1=-c\dot{x}_1-kx+u+d,
\end{align*}
with mass $m=1$, uncertain damping constant $c\in[0.1,0.3]$, uncertain spring constant $k\in[0.5,1.5]$ and additive disturbances $|d_t|\leq 0.2$. 
The true \textit{unkown} parameters are $c^*=0.3$, $k^*=0.5$. 
The state is defined as $x=(x_1;\dot{x}_1)\in\mathbb{R}^2$.
We consider the constraint set $\mathcal{Z}=[-0.1,1.1]\times [-5,5]\times[-5,5]$ and use a Euler discretization with a sampling time of $T_s=0.1~s$. 
The control goal is to alternately track the origin and the setpoint $x_s=(1;0)$.

\subsubsection*{Offline Computation} 
The matrices $P,K$ are computed using the LMIs~\eqref{eq:LMIs} in Appendix~\ref{App:tube}, such that $P$ satisfies~\eqref{eq:P} with  $Q=\text{diag}(1,10^{-2})$, $R=10^{-1}$ and is contractive with $A_{cl,\theta}^\top P A_{cl,\theta}\leq \rho^2P$ and $\rho=0.75$.
The polytope $\mathcal{P}$ is computed as the maximal $\rho$-contractive set~\cite{pluymers2005efficient,blanchini2008set} for the constraint set
\begin{align*}
\tilde{\mathcal{Z}}=\{|x_1|\leq 0.1, |x_2|\leq 5,u\in[-5,4] \},
\end{align*}
 which is described by $r=18$ linear inequalities.  
Note, that the set $\tilde{\mathcal{Z}}$ is chosen such that $(x_s,u_{s,\overline{\theta}_0})\oplus\tilde{\mathcal{Z}}\subseteq\mathcal{Z}$ for both setpoints $x_s\in\{(0,0),(1,0)\}$.
Since $x_s\neq0$, the terminal set cannot be directly computed using Proposition~\ref{prop:term}, due to the additional uncertainty at the steady state and the parameter dependence of the steady state input $u_{s,\theta}=kx_s\in[0.5,1.5]$.
However, in the considered example, we can choose the terminal set 
\begin{align*}
\mathcal{X}_f=\{(x,s)\in\mathbb{R}^{n+1}|~s+H_i(x-x_s)\leq 1\}.
\end{align*}
The RPI condition~\eqref{eq:term_cond} in Proposition~\ref{prop:term} changes to
\begin{align*}
1-\rho_{\overline{\theta}_0}+\eta_0L_{\mathbb{B}}+w_{\eta_0}(x_s,u_s) + c_{\max}\overline{d}\leq 1,
\end{align*}
which is satisfied with $c_{\max}=1$, $\overline{d}=0.0582$, $\rho_{\overline{\theta}_0}=0.75$, $\eta_0 L_{\mathbb{B}}=0.0363$, $w_{\eta_0}(x_s,u_s)=0.1455$. 
Additional details regarding the terminal set for steady states $x_s\neq0$ can be found in Proposition~\ref{prop:term_app_2} in the appendix.

\subsubsection*{Closed-loop performance improvement}
We implement the proposed approach (Adaptive RMPC) with a prediction horizon of $N=14$ and a window length of $M=10$. 
For comparison, we also implement a purely robust formulation (RMPC) without any model adaptation, which corresponds to the robust MPC scheme in~\cite{Robust_TAC_19}. 
The corresponding closed-loop performance can be seen in Figure~~\ref{fig:close}. 
The parameter update significantly improves the tracking error.
The complexity of the MPC optimization problem~\eqref{eq:MPC} is not affected by the online parameter adaptation. 
The only increase in computational complexity, relative to robust MPC, is the computation of the hypercube overapproximation (Alg.~\ref{alg:HC}), which approximately increases the computation time by $2\%$. 
Thus,  combining the robust MPC approach in~\cite{Robust_TAC_19} with online parameter adaptation significantly improves the closed-loop performance with a marginally increase in computational complexity. 
\begin{figure}[!htbp]
      \centering
      \includegraphics[width=0.48\textwidth]{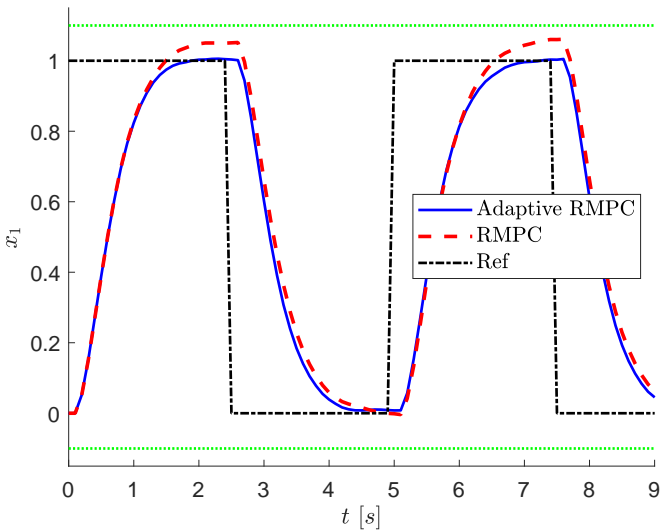}
      \caption{Comparison of closed loop trajectories with setpoint changing between 0 and 1 for the proposed adaptive RMPC and the RMPC~\cite{Robust_TAC_19}.}
      \label{fig:close}
\end{figure}

\subsubsection*{Parameter estimation}
The parameter estimation can be seen in Figure~\ref{fig:param}. 
The hypercube parameter set $\Theta_t^{HC}$ is shrinking during online operation and the LMS point estimate $\hat{\theta}_{t}$  converges to a small neighborhood of the true parameters $\theta^*=(1,-1)$.  
The projection of the LMS point estimate on the parameter set $\Theta_t^{HC}$ comes at virtually no cost, but can significantly improve the transient error, especially in case of large disturbances $d_t$. 
\begin{figure}[!htbp]
      \centering
      \includegraphics[width=0.5\textwidth]{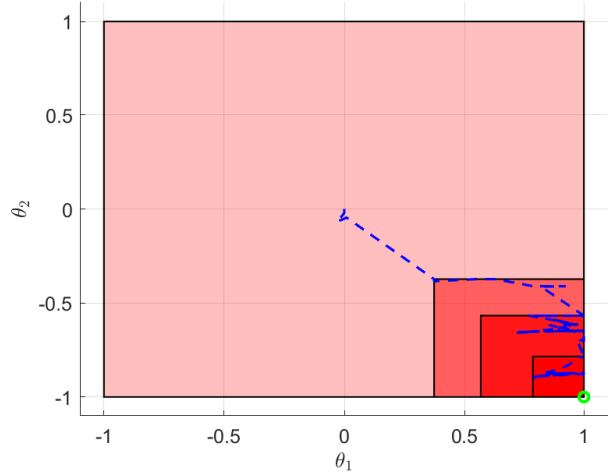}
      \caption{Shrinking parameter set $\Theta^{HC}_t$ (red) at time steps $t=\{0,10,26,90\}$, evolution of the LMS estimate $\hat{\theta}_t$ (blue dashed) and true parameters $\theta^*=(1,-1)$ (green dot).}
      \label{fig:param}
\end{figure}
 
\subsubsection*{Complexity robust MPC}
In the following we compare the conservatism and computational complexity of the different robust MPC formulations discussed in Remark~\ref{rk:robust_methods}. 
To this end, we consider the predicted trajectory $x^*_{\cdot|t}$ of the proposed approach at $t=0$ and compute\footnote{%
To allow for an intuitive comparison, we centered the homothetic and flexible tube (\cite{lorenzen2019robust,Lu2019RAMPC}) around the trajectory $\overline{x}^*_{\cdot|t}$, even though this may introduce additional conservatism. 
} the tube size $s_{\cdot|t}$ for the different formulations.
The corresponding result can be seen in Figure~\ref{fig:tube_size} and Table~\ref{tab:compare}. 
The number of optimization variables are displayed for a condensed formulation with the dynamic equality constraints eliminated.

The tube size $s$ of the considered formulation is approximately  $5\%$ larger than the flexible tube~\cite{Lu2019RAMPC} approach, while the proposed approach requires only $11\%$ of the number of optimization variables. 
This effect is amplified in comparison to the  homothetic tube formulation~\cite{lorenzen2019robust}, where the consider formulation results in an approximately $16\%$ larger tube size $s$, while the number of optimization variables and constraints are drastically reduced.  
In the considered scenario, the tube size computed using the formula~\eqref{eq:w_1} is equivalent to the proposed approach~\eqref{eq:w_2}, which implies that the conservatism of using $\overline{d}\geq \overline{d}_i$ is negligible in this example.  
The simplified formula~\eqref{eq:w_3} is significantly more conservative, resulting in approximately $3$-times the tube size $s$.  

In general, there exists a degree of freedom in the choice of the robust MPC formulation, that allows a user to trade conservatism vs. computational complexity.
These trade-offs between complexity and conservatism are, however, very problem specific. 
In this example the uncertainty is mainly due to the parametric uncertainty, which is why \eqref{eq:w_2} and \eqref{eq:w_1} are approximately equivalent and \eqref{eq:w_3}, which overapproximates the effect of the parametric uncertainty, is so conservative. 

\begin{table}[H]
\centering
\begin{tabular}{|c|r|r|c|}
\hline
Approach&$\#$ opt. var.&$\#$ ineq. con.&$s_{N|t}$ 		\\ \hline
Homothetic tube~\eqref{eq:homothetic_tube}~\cite{lorenzen2019robust}& $18173$ &$18228$&0.75 \\ \hline
Flexible tube~\eqref{eq:Lu_19}~\cite{Lu2019RAMPC}&$266$&1092&0.83\\\hline 
Proposed $w_1$~\eqref{eq:w_1}~\cite{Robust_TAC_19}&$30$&$1092$&$0.87$\\ \hline
Proposed $w_2$~\eqref{eq:w_2}&$30$&$1092$ &$0.87$\\ \hline
Simplified $w_3$~\eqref{eq:w_3}&$30$&$336$&$2.48$	\\ \hline
Nominal MPC&$14$&$84$&-\\\hline
\end{tabular}
\caption{Uncertainty propagation using robust tube approaches - Computational complexity and conservatism}
\label{tab:compare}
\end{table}
\begin{figure}[!htbp]
      \centering
      \includegraphics[width=0.5\textwidth]{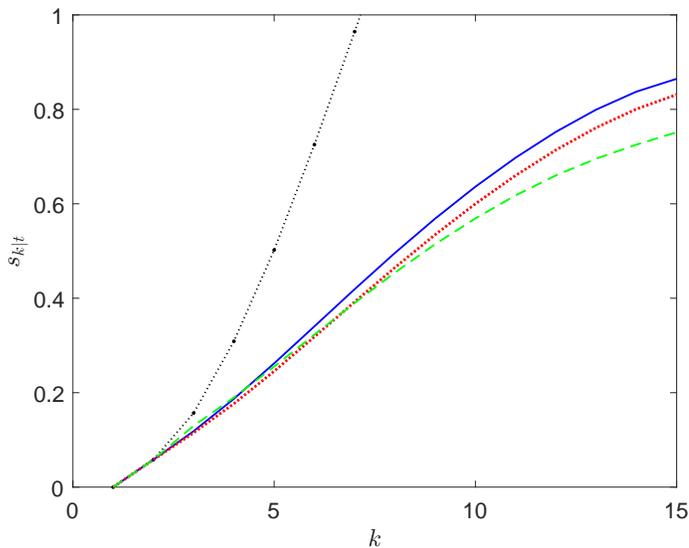}
      \caption{Tube size: Proposed robust formulation~\eqref{eq:w_2} (black, dash-dott), simplified formula~\eqref{eq:w_3} (red, dotted), flexible tube~\eqref{eq:Lu_19} (red, dotted), homothetic tube~\eqref{eq:homothetic_tube} (green, dashed).}
      \label{fig:tube_size}
\end{figure}

\section{Conclusion}
\label{sec:sum}
We have presented a robust adaptive MPC scheme for linear uncertain systems that ensures robust
constraint satisfaction, recursive feasibility and $\mathcal{L}_2$ stability. 
The proposed scheme improves the performance and reduces the conservatism online using parameter adaptation and set membership estimation. 
In addition, the formulations are such that the computational complexity is constant during runtime and only moderately increased compared to a nominal MPC scheme.
The trade-off between computational complexity and conservatism regarding different MPC formulations has been investigated with a numerical example. 
Current research is focused on reducing the conservatism and extending the framework to nonlinear uncertain systems~\cite{kohler2019robust} as a competing approach to \cite{adetola2011robust}. 
\bibliographystyle{IEEEtran}  
\bibliography{Literature_short}  

\begin{thebibliography}{10}
\providecommand{\url}[1]{#1}
\csname url@samestyle\endcsname
\providecommand{\newblock}{\relax}
\providecommand{\bibinfo}[2]{#2}
\providecommand{\BIBentrySTDinterwordspacing}{\spaceskip=0pt\relax}
\providecommand{\BIBentryALTinterwordstretchfactor}{4}
\providecommand{\BIBentryALTinterwordspacing}{\spaceskip=\fontdimen2\font plus
\BIBentryALTinterwordstretchfactor\fontdimen3\font minus
  \fontdimen4\font\relax}
\providecommand{\BIBforeignlanguage}[2]{{%
\expandafter\ifx\csname l@#1\endcsname\relax
\typeout{** WARNING: IEEEtran.bst: No hyphenation pattern has been}%
\typeout{** loaded for the language `#1'. Using the pattern for}%
\typeout{** the default language instead.}%
\else
\language=\csname l@#1\endcsname
\fi
#2}}
\providecommand{\BIBdecl}{\relax}
\BIBdecl

\bibitem{Koehler2019Adaptive}
J.~K\"{o}hler, E.~Andina, R.~Soloperto, M.~A. M{\"u}ller, and F.~Allg{\"o}wer,
  ``Linear robust adaptive model predictive control: Computational complexity
  and conservatism,'' in \emph{Proc.\ 58th IEEE Conf.\ Decision and Control
  (CDC)}, 2019, pp. 1383--1388.

\bibitem{rawlings2017model}
J.~B. Rawlings, D.~Q. Mayne, and M.~Diehl, \emph{Model Predictive Control:
  Theory, Computation, and Design}.\hskip 1em plus 0.5em minus 0.4em\relax Nob
  Hill Pub., 2017.

\bibitem{kouvaritakis2016model}
B.~Kouvaritakis and M.~Cannon, \emph{Model predictive control}.\hskip 1em plus
  0.5em minus 0.4em\relax Springer, 2016.

\bibitem{fleming2015robust}
J.~Fleming, B.~Kouvaritakis, and M.~Cannon, ``Robust tube {MPC} for linear
  systems with multiplicative uncertainty,'' \emph{IEEE Trans. Autom. Control},
  vol.~60, no.~4, pp. 1087--1092, 2015.

\bibitem{Robust_TAC_19}
J.~K\"{o}hler, R.~Soloperto, M.~A. M{\"u}ller, and F.~Allg{\"o}wer, ``A
  computationally efficient robust model predictive control framework for
  uncertain nonlinear systems,'' submitted to IEEE Transactions on Automatic
  Control, 2019, arXiv preprint arXiv:1910.12081.

\bibitem{rakovic2012homothetic}
S.~V. Rakovi{\'c}, B.~Kouvaritakis, R.~Findeisen, and M.~Cannon, ``Homothetic
  tube model predictive control,'' \emph{Automatica}, vol.~48, no.~8, pp.
  1631--1638, 2012.

\bibitem{aswani2013provably}
A.~Aswani, H.~Gonzalez, S.~S. Sastry, and C.~Tomlin, ``Provably safe and robust
  learning-based model predictive control,'' \emph{Automatica}, vol.~49, no.~5,
  pp. 1216--1226, 2013.

\bibitem{di2016indirect}
S.~Di~Cairano, ``Indirect adaptive model predictive control for linear systems
  with polytopic uncertainty,'' in \emph{Proc. American Control Conf.\
  (ACC)}.\hskip 1em plus 0.5em minus 0.4em\relax IEEE, 2016, pp. 3570--3575.

\bibitem{tanaskovic2014adaptive}
M.~Tanaskovic, L.~Fagiano, R.~Smith, and M.~Morari, ``Adaptive receding horizon
  control for constrained {MIMO} systems,'' \emph{Automatica}, vol.~50, no.~12,
  pp. 3019--3029, 2014.

\bibitem{tanaskovic2019adaptive}
M.~Tanaskovic, L.~Fagiano, and V.~Gligorovski, ``Adaptive model predictive
  control for linear time varying {MIMO} systems,'' \emph{Automatica}, vol.
  105, pp. 237--245, 2019.

\bibitem{bujarbaruah2018adaptive}
M.~Bujarbaruah, X.~Zhang, and F.~Borrelli, ``Adaptive {MPC} with chance
  constraints for {FIR} systems,'' in \emph{Proc. American Control Conf.\
  (ACC)}.\hskip 1em plus 0.5em minus 0.4em\relax IEEE, 2018, pp. 2312--2317.

\bibitem{lorenzen2019robust}
M.~Lorenzen, M.~Cannon, and F.~Allg{\"o}wer, ``Robust {MPC} with recursive
  model update,'' \emph{Automatica}, vol. 103, pp. 461--471, 2019.

\bibitem{adetola2011robust}
V.~Adetola and M.~Guay, ``Robust adaptive {MPC} for constrained uncertain
  nonlinear systems,'' \emph{Int. J. Adapt. Control Signal Process.}, vol.~25,
  no.~2, pp. 155--167, 2011.

\bibitem{marafioti2014persistently}
G.~Marafioti, R.~R. Bitmead, and M.~Hovd, ``Persistently exciting model
  predictive control,'' \emph{Int. J. Adapt. Control Signal Process.}, vol.~28,
  no.~6, pp. 536--552, 2014.

\bibitem{heirung2017dual}
T.~A.~N. Heirung, B.~E. Ydstie, and B.~Foss, ``Dual adaptive model predictive
  control,'' \emph{Automatica}, vol.~80, pp. 340--348, 2017.

\bibitem{mesbah2018stochastic}
A.~Mesbah, ``Stochastic model predictive control with active uncertainty
  learning: {A} survey on dual control,'' \emph{Annual Reviews in Control},
  vol.~45, pp. 107--117, 2018.

\bibitem{RS_CDC_Dual_19}
R.~Soloperto, J.~K\"{o}hler, M.~A. M{\"u}ller, and F.~Allg{\"o}wer, ``Dual
  adaptive {MPC} for output tracking of linear systems,'' in \emph{Proc.\ 58th
  IEEE Conf.\ Decision and Control (CDC)}, 2019, pp. 1377--1382.

\bibitem{Lu2019RAMPC}
X.~Lu and M.~Cannon, ``Robust adaptive tube model predictive control,'' in
  \emph{Proc. American Control Conf.\ (ACC)}, 2019, pp. 3695--3701.

\bibitem{Elisa}
E.~Andina, ``Complexity and conservatism in linear robust adaptive model
  predictive control,'' Master's thesis, University of Stuttgart, 2019.

\bibitem{chisci1998block}
L.~Chisci, A.~Garulli, A.~Vicino, and G.~Zappa, ``Block recursive
  parallelotopic bounding in set membership identification,''
  \emph{Automatica}, vol.~34, no.~1, pp. 15--22, 1998.

\bibitem{lennart1999system}
L.~Ljung, ``System identification: theory for the user,'' \emph{PTR Prentice
  Hall, Upper Saddle River, NJ}, pp. 1--14, 1999.

\bibitem{blanchini1999set}
F.~Blanchini, ``Set invariance in control,'' \emph{Automatica}, vol.~35,
  no.~11, pp. 1747--1767, 1999.

\bibitem{kouramas2005minimal}
K.~I. Kouramas, S.~V. Rakovic, E.~C. Kerrigan, J.~Allwright, and D.~Q. Mayne,
  ``On the minimal robust positively invariant set for linear difference
  inclusions,'' in \emph{Proc.\ 44th IEEE Conf.\ Decision and Control (CDC)},
  2005, pp. 2296--2301.

\bibitem{rakovic2005minimal}
S.~Rakovic, K.~Kouramas, E.~Kerrigan, J.~Allwright, and D.~Mayne, ``The minimal
  robust positively invariant set for linear difference inclusions and its
  robust positively invariant approximations,'' \emph{Automatica}, 2005.

\bibitem{pluymers2005efficient}
B.~Pluymers, J.~Rossiter, J.~Suykens, and B.~De~Moor, ``The efficient
  computation of polyhedral invariant sets for linear systems with polytopic
  uncertainty,'' in \emph{Proc. American Control Conf.\ (ACC)}.\hskip 1em plus
  0.5em minus 0.4em\relax IEEE, 2005, pp. 804--809.

\bibitem{blanchini2008set}
F.~Blanchini and S.~Miani, \emph{Set-theoretic methods in control}.\hskip 1em
  plus 0.5em minus 0.4em\relax Springer, 2008.

\bibitem{kerrigan2001robust}
E.~C. Kerrigan, ``Robust constraint satisfaction: Invariant sets and predictive
  control,'' Ph.D. dissertation, University of Cambridge, 2001.

\bibitem{fleming2016robust}
J.~Fleming, ``Robust and stochastic {MPC} of uncertain-parameter systems,''
  Ph.D. dissertation, University of Oxford, 2016.

\bibitem{rakovic2013homothetic}
S.~V. Rakovi{\'c} and Q.~Cheng, ``Homothetic tube {MPC} for constrained linear
  difference inclusions,'' in \emph{Proc.\ 25th Chinese Control and Decision
  Conf. (CCDC)}, 2013, pp. 754--761.

\bibitem{zhang2019computationally}
K.~Zhang, C.~Liu, and Y.~Shi, ``Computationally efficient adaptive model
  predictive control for constrained linear systems with parametric
  uncertainties,'' in \emph{Proc. 28th International Symposium on Industrial
  Electronics (ISIE)}.\hskip 1em plus 0.5em minus 0.4em\relax IEEE, 2019, pp.
  2152--2157.

\bibitem{kohler2019robust}
J.~K{\"o}hler, P.~K{\"o}tting, R.~Soloperto, F.~Allg{\"o}wer, and M.~A.
  M{\"u}ller, ``A robust adaptive model predictive control framework for
  nonlinear uncertain systems,'' \emph{arXiv preprint arXiv:1911.02899}, 2019.

\bibitem{limon2008mpc}
D.~Lim{o}n, I.~Alvarado, T.~Alamo, and E.~F. Camacho, ``{MPC} for tracking
  piecewise constant references for constrained linear systems,''
  \emph{Automatica}, vol.~44, pp. 2382--2387, 2008.

\end{thebibliography}
\cleardoublepage
\appendix
\label{sec:app}
Appendix~\ref{App:tube} discusses the offline computation of $P,K,\mathcal{P}$ and proves some auxiliary results. 
Appendix~\ref{App:terminal} extends the result in Proposition~\ref{prop:term} for the terminal ingredients to steady states $x_s\neq 0$. 
\subsection{Polytopic tube}
\label{App:tube}
In the following, we discuss the design of the polytopic tube in more details and prove some auxiliary results.  
\subsubsection*{LMIs}
Assumption~\ref{ass:stable} requires matrices $P,K$ that satisfy the Lyapunov inequality~\eqref{eq:P} for any parameters $\theta\in\Theta_0^{HC}$. 
Additionally, it may be advantageous to ensure that the sublevel set $\|x\|_P^2\leq 1$ is contractive, satisfies the constraints and is RPI, compare also the LMIs in the numerical example in~\cite{Robust_TAC_19}.      
The following semidefinite programming (SDP) can be used to enforce these conditions with $P=X^{-1}$, $K=YP$, the vertices $\theta^j=\overline{\theta}_0+\eta_0\tilde{e}_l$, $d^k$, and a scalar $\lambda\geq 0 $ (from the S procedure)
\begin{subequations}
\label{eq:LMIs}
\begin{align}
&\min_{X,Y}-\log\det(X),\\
\label{eq:LMI_term}
&\begin{pmatrix}
X &(A_{\theta^j}X+B_{\theta^j}Y)^\top&Q^{1/2}X&R^{1/2}Y^\top\\
*&X&0&0\\
*&*&I&0\\
*&*&*&I
\end{pmatrix}\geq 0.\\
\label{eq:contract}
&\begin{pmatrix}
\rho X &(A_{\theta^j}X+B_{\theta^j}Y)^\top\\
*&\rho X
\end{pmatrix}\geq 0.\\
\label{eq:constraints}
&\begin{pmatrix}
1&F_jX+G_jY\\
*&X
\end{pmatrix}\geq 0,\\
\label{eq:RPI}
&\begin{pmatrix}
\lambda X&0&(A_{\theta^j}X+B_{\theta^j}Y)^\top\\
0&1-\lambda&(d^k)^\top\\
*&*&X
\end{pmatrix}\geq 0.
\end{align}
\end{subequations}
Note that the constraints in~\eqref{eq:LMIs} are only LMIs for a fixed given $\rho,\lambda$, which can be adjusted in an outer loop (similar to bisection).
Given $P,K$, standard methods can be used to compute the polytope $\mathcal{P}$, such as computing the minimal RPI set~\cite{kouramas2005minimal,rakovic2005minimal}, the maximal invariant/contractive set~\cite{pluymers2005efficient,blanchini2008set} or maximal RPI set~\cite{kerrigan2001robust},\cite[Alg.~4.3]{fleming2016robust}. 
To the best of the authors knowledge, there exists no algorithm to compute a polytope $\mathcal{P}$, that is explicitly tailored to satisfying inequality~\eqref{eq:term_cond}.

\subsubsection*{Alternative tube propagation}
In Remark~\ref{rk:robust_methods} various alternative robust tube methods are discussed. 
In the following, we briefly derive the formula~\eqref{eq:w_3} and show that it can equally be used in the proposed scheme. 
Suppose the matrix $B_{\theta}$ is given by $B_{\theta}=B_0+\sum_{i=1}^{p_B} [\theta]_i B_i$,  where $p_B<p$ denotes the number of parameters with $B_i\neq 0$. 
Proposition~\ref{prop:cont_w} ensures that the following inequality holds
\begin{align*}
w_{\eta}(x,u)\leq \tilde{w}_{\eta}(x,u):=w_{\eta}(0,u-Kx)+\eta L_{\mathbb{B}}\max_i H_ix.
\end{align*}
Thus, we can use $\tilde{w}_{\eta}$ instead of $w_{\eta}$ to upper bound the uncertainty. 
Note that this function also satisfies inequality~\eqref{eq:Lipschitz}: 
\begin{align*}
&\tilde{w}_{\eta}(x,v+K(x-z))\\
=&w_{\eta}(0,v+K(x-z)-Kx)+\eta L_{\mathbb{B}}\max_i H_i x \\
\leq&w_{\eta}(0,v-Kz)+\eta L_{\mathbb{B}}(\max_i H_i z+\max_i H_i(x-z))\\
=&\tilde{w}_{\eta}(z,v)+\eta L_{\mathbb{B}}\max_i H_i(x-z).
\end{align*}
In Theorem~\ref{thm:main} we use the fact that the function $w_{\eta}$ is an upper bound on the uncertainty, satisfies inequality~\eqref{eq:Lipschitz} and is linear in $\eta$.  
Since the function $\tilde{w}_{\eta}$ satisfies the same properties, we can also use this function to reduce the computational demand.
The tube propagation~\eqref{eq:w_3} can be implemented using
\begin{align*}
s^+\geq&(\rho_{\theta}+\eta L_{\mathbb{B}})s+\eta L_{\mathbb{B}}g+\eta H_i(D(0,u-Kx))\tilde{e}_l+\overline{d},\\
g\geq& H_kx,~k=1,\dots,r,~l=1,\dots,2^{p_B},~i=1,\dots,r,
\end{align*}
with the $2^{p_B}$ vertices $\tilde{e}_l$ of the reduced hypercube in $\mathbb{R}^{p_B}$. 
This implementation uses $r(1+2^{p_B})$ linear inequality constraints and an additional auxiliary variable $g$. 
Thus, in case $p_B<<p$, this can lead to a significant reduction in the computational complexity. 
Furthermore, for $p_B=0$, the auxiliary variable $g$ is not necessary and the constraint can be directly implemented using~\eqref{eq:w_3} with $r$ linear inequality constraints. 
Compared to~\eqref{eq:w_2} this formulation can be significantly more conservative, compare numerical example.

\subsection{Terminal ingredients - setpoint tracking}
\label{App:terminal}
In the following, we consider more general terminal ingredients for steady states other than the origin. 
The additional difficulties are as follows: a) the uncertainty at the steady state $(x_s,u_s)$ is greater than the uncertainty at the origin; b) not every steady state remains a steady state if the parameters $\theta$ change and c) the corresponding steady state input $u_s$ may change if $\theta$ changes. 
To this end, the following assumption generalizes the conditions in Assumption~\ref{ass:term}. 
\begin{assumption}\label{ass:term_2}
Consider the set $\Theta_0$ from Assumption~\ref{ass:model}.
There exists a terminal region $\mathcal{X}_{f,\Theta^{HC}}\subset\mathbb{R}^{n+1}$ and a terminal controller $k_{f,\Theta^{HC}}:\mathbb{R}^n\rightarrow\mathbb{R}^m$, such that the following properties hold 
\begin{description}
\item[] for all $\Theta^{HC}=\bar{\theta}\oplus\eta\mathbb{B}_p\subseteq\Theta_0^{HC}$
\item[] for all $\tilde{\Theta}^{HC}=\tilde{\bar{\theta}}\oplus\tilde{\eta}\mathbb{B}_p\subseteq\Theta^{HC}$
\item[] for all $(x,s)\in\mathcal{X}_{f,\Theta^{HC}}$,~$u=k_{f,\Theta^{HC}}(x)$
\item[] for all $\tilde{s}\in\mathbb{R}_{\geq0}$,
\item[] for all $s^+\in [0,(\rho_{\overline{\theta}}+\eta L_{\mathbb{B}})s+w_{\eta}(x,u)+\overline{d}-\tilde{s}]$
\item[] for all $x^+\;\textrm{s.t.}\;\max_{i}H_i(x^+-A_{\overline{\theta}}x-B_{\overline{\theta}}u)\leq \tilde{s}$:
\end{description}
\begin{subequations}
\begin{align}
\label{eq:ass_term_2_2}
	&(x^+,s^+)\in\mathcal{X}_{f,\tilde{\Theta}^{HC}},\\
\label{eq:ass_term_3_2}
	&F_jx+G_ju+c_js\leq 1, \quad j=1,...,q.
\end{align}
\end{subequations}
\end{assumption}

The following proposition provides a simple terminal set for this case, under additional conditions on the steady state and possibly conservative bounds.
\begin{proposition}
\label{prop:term_app_2}
Let Assumption~\ref{ass:model} and \ref{ass:stable} hold.
Suppose there exists a state $x_s$, such that for any ${\theta}\in\Theta_0^{HC}$, there exists an input $u_{s,{\theta}}$ that satisfies $x_s=A_{{\theta}}x_s+B_{{\theta}}u_{s,{\theta}}$. 
Given parameters $\theta\in\mathbb{R}^p$ and a parameter set $\Theta^{HC}=\overline{\theta}\oplus\eta\mathbb{B}_p$, we define
\begin{subequations}
\begin{align}
\label{eq:f_theta}
f_{\theta}:=&\min_j\dfrac{1-F_j x_s-G_j u_{s,\theta}}{c_j},\\
\label{eq:w_theta}
{w}_{\theta}:=&\max_{i,j} H_i D(x_s,u_{s,{\theta}})\tilde{e}_j,\\
\label{eq:f_bar}
\underline{f}_{\Theta^{HC}}:=&\min_{\theta\in\Theta^{HC}}f_{\theta},\\
\label{eq:w_bar}
\overline{w}_{\Theta^{HC}}:=&\max_{\theta\in\Theta^{HC}}w_{\theta}.
\end{align}
\end{subequations}
Suppose that the following condition holds
\begin{align}
\label{eq:term_cond_app_2}
\eta_0 \overline{w}_{\Theta^{HC}_0}+\overline{d}{\leq} \underline{f}_{\Theta_0^{HC}}(1-\rho_{\overline{\theta}_0}-\eta_0 L_{\mathbb{B}}).
\end{align}
Then the terminal set 
\begin{align}
\label{eq:term_set_app_2}
\mathcal{X}_{f,\Theta^{HC}}:=&\{(x,s)\in\mathbb{R}^{n+1}|s+H_i (x-x_s)\leq \underline{f}_{\Theta^{HC}}\},
\end{align}
and the terminal controller $k_{f,\Theta^{HC}}(x):=u_{s,\overline{\theta}}+K(x-x_s)$
satisfy Assumption~\ref{ass:term_2}. 
\end{proposition}
\begin{proof}
As in Proposition~\ref{prop:term},  $\Theta^{HC}\subseteq\Theta^{HC}_0$, and Proposition~\ref{prop:cont_rho} imply
	$\rho_{\overline{\theta}}+\eta L_{\mathbb{B}}\leq\rho_{\overline{\theta}_{0}}+\eta_{0}L_{\mathbb{B}}$.
Furthermore, $\Theta^{HC}\subseteq\Theta_0^{HC}$ and the definitions in~\eqref{eq:f_bar},~\eqref{eq:w_bar} 
ensure that $\overline{w}_{\Theta^{HC}}\leq \overline{w}_{\Theta_0^{HC}}$ and $\underline{f}_{\Theta^{HC}}\geq \underline{f}_{\Theta_0^{HC}}$.  
Thus, satisfaction of~\eqref{eq:term_cond_app_2}, ensures satisfaction of~\eqref{eq:term_cond_app_2} with $\eta_0,\overline{\theta}_0,\Theta_0^{HC}$ replaced by $\eta,\overline{\theta},\Theta^{HC}$. 
Satisfaction of the tightened constraints~\eqref{eq:ass_term_3_2} follows from 
\begin{align*}
	&F_jx+G_ju+c_js\\
=&[F+GK]_j(x-x_s)+F_jx_s+G_ju_{s,\overline{\theta}}+c_j s\\
\stackrel{\eqref{eq:c_j}}{\leq} &c_j \max_i H_i(x-x_s)+c_j s+F_j x_s+G_j u_{s,\overline{\theta}}\\
\stackrel{\eqref{eq:f_theta},\eqref{eq:term_set_app_2}}{\leq} &c_j\underline{f}_{\Theta^{HC}}+1-c_j{f}_{\overline{\theta}}\stackrel{\eqref{eq:f_bar}}{\leq} 1.
\end{align*}
We use the following bound based on Proposition~\ref{prop:cont_w}:
\begin{align}
\label{eq_bound_w_tilde_terminal_2}
&	\max_{(x,s)\in\mathcal{X}_{f,\Theta^{HC}}}{w}_{\eta}(x,k_{f,\Theta^{HC}}(x))+\eta L_{\mathbb{B}}s\nonumber\\
\stackrel{\eqref{eq:Lipschitz}}{\leq}&{w}_{\eta}(x_s,u_{s,\overline{\theta}})+\eta L_{\mathbb{B}}\max_{(x,s)\in\mathcal{X}_f}(\max_i H_i (x-x_s)+s)\nonumber\\
\stackrel{\eqref{eq:term_set_app_2}}{\leq} &{w}_{\eta}(x_s,u_{s,\overline{\theta}})+\eta L_{\mathbb{B}}\underline{f}_{\Theta^{HC}}\nonumber\\\stackrel{\eqref{eq:w_theta}\eqref{eq:w_bar}}{\leq} &\eta \overline{w}_{\Theta^{HC}}+\eta L_{\mathbb{B}}\underline{f}_{\Theta^{HC}}.
\end{align}
The state $x^+$ can be bounded as follows
\begin{align}
\label{eq:bound_terminal_x_plus_2}
	\max_{i}H_i(x^+-x_s)\leq&\max_{i}H_iA_{cl,\overline{\theta}}(x-x_s)+\tilde{s}\nonumber\\
\stackrel{\eqref{eq:rho}}{\leq}&\rho_{\bar{\theta}}\max_{i}H_i(x-x_s)+\tilde{s}.
\end{align}
The RPI condition~\eqref{eq:ass_term_2_2} follows with
\begin{align*}
&s^+ + H_i (x^+-x_s)\\
\leq&\rho_{\overline{\theta}}s+\eta  L_{\mathbb{B}}s+w_{\eta}(x,k_{f,\Theta^{HC}}(x))+\overline{d}-\tilde{s}\\
&+ H_i A_{cl,\overline{\theta}}(x-x_s)+\tilde{s}\\
\stackrel{\eqref{eq_bound_w_tilde_terminal_2}\eqref{eq:bound_terminal_x_plus_2}}{\leq}& \rho_{\overline{\theta}}(s+\max_i H_i (x-x_s))+\eta \overline{w}_{\Theta^{HC}}+\eta L_{\mathbb{B}}\underline{f}_{\Theta^{HC}}+\overline{d}\\
\stackrel{\eqref{eq:term_set_app_2}}{\leq}& \rho_{\overline{\theta}}\underline{f}_{\Theta^{HC}}+\eta L_{\mathbb{B}}\underline{f}_{\Theta^{HC}}+\eta\overline{w}_{\Theta^{HC}}+\overline{d}\stackrel{\eqref{eq:term_cond_app_2}}{\leq} \underline{f}_{\Theta^{HC}}. 
\end{align*}
\end{proof}
Since this terminal set satisfies all the properties of the terminal set used in the proof of Theorem~\ref{thm:main}, we can use this design to ensure robust recursive feasibility, even if the terminal set is not centered around the origin. 
\begin{remark}
In principle, the terminal constraint could be replaced by the parameter dependent, less conservative constraint:
\begin{align}
\mathcal{X}_{f,\overline{\theta}}:=&\{(x,s)\in\mathbb{R}^{n+1}|s+H_i (x-x_s)\leq  f_{\overline{\theta}}\}.
\end{align}
However, this complicates the analysis of recursive feasibility under changing parameters and is thus not pursued here.

For many systems $u_{s,\theta}$ is affine in $\theta$ and thus $\underline{f}_{\Theta^{HC}}$ can be efficiently computed using linear programs. 
By noting that $1/\underline{f}_{\Theta^{HC}}$ plays the same role as $c_{\max}$ in~\eqref{eq:term_cond}, we can see that the main difference to the condition in Prop.~\ref{prop:term} is the additional disturbance term $\eta\overline{w}_{\Theta^{HC}}$ due to the effect of the parametric uncertainty at the steady state.  
In case $B_{\theta}$ is independent of $\theta$, $w_{\theta}$ is independent of $\theta$ and $\overline{w}_{\Theta^{HC}}$ can be computed using a linear program.

Furthermore, in case condition~\eqref{eq:term_cond_app_2} is not satisfied since the desired setpoint is too close to the constraints or the uncertainty at the setpoint $\eta\overline{w}_{\Theta^{HC}}$ is too large, a natural solution is to consider an artificial steady-state as in~\cite{limon2008mpc}. 
In this case the scheme will initially stabilize a setpoint that has a larger distance to the constraints. 
By reducing the uncertainty ($\eta_t$) in closed-loop operation, this setpoint can then potentially move closer to the constraint set.
In general, combing online adaptation with artificial setpoints seems like a promising approach to address practical challenges, compare e.g. also~\cite{RS_CDC_Dual_19}. 
Deriving a corresponding more general formulation is an open issue. 
\end{remark}

\end{document}